\documentclass[aps,prxquantum,twocolumn,amsmath,amssymb,superscriptaddress]{revtex4-2}
\usepackage{graphicx}
\usepackage{dcolumn}
\usepackage{bm}
\usepackage{hyperref}
\usepackage{subfigure}

\usepackage{epstopdf}
\usepackage{tikz}

\usepackage{subfigure}
\usepackage{dcolumn}
\usepackage{bm,bbm}  
\usepackage{amsmath}

\usepackage{amssymb}
\usepackage{mathrsfs}
\usepackage[latin1]{inputenc}
\usepackage{pstricks,pst-node,pst-text,pst-3d}

\usepackage[all]{xy}

\usepackage{setspace}
\usepackage{tikz}
\usetikzlibrary{shapes,arrows}
\usepackage{lipsum}
\usepackage{lmodern}
\usepackage{tcolorbox}

\usepackage{diagbox}
\usepackage[ ruled ]{algorithm2e}
\usepackage{algpseudocode}
\usepackage{textcomp}
\usepackage{graphicx}
\usepackage{subfigure}

\makeatletter
\newcommand\figcaption{\def\@captype{figure}\caption}
\newcommand\tabcaption{\def\@captype{table}\caption}
\makeatother

\newcommand{\tr}{{\rm Tr}}

\newtheorem{Assumption}{Assumption}
\newtheorem{Theorem}{Theorem}

\newtheorem{Corollary}{Corollary}
\newtheorem{Definition}{Definition}
\newtheorem{Conjecture}{Conjecture}

\newtheorem{Lemma}{Lemma}

\definecolor{maroon}{rgb}{0.7,0,0}

\definecolor{ngreen}{rgb}{0.3,0.7,0.3}

\definecolor{golden}{rgb}{0.8,0.6,0.1}

\begin{document}
	
\title{Optimization landscapes of variational quantum algorithms }
\author{Xiaozhen Ge}
\affiliation{Department of Automation, Tsinghua University, Beijing, 100084, China}
\affiliation{Department of Applied Mathematics, The Hong Kong Polytechnic University, kowloon, 999077, Hong Kong, China}
\affiliation{State Key Laboratory of Autonomous Intelligent Unmanned Systems, Shanghai Research Institute for Intelligent Autonomous Systems, Tongji University, Shanghai, 201203, China}

\author{Shuming Cheng}
\email{drshuming.cheng@gmail.com}
\affiliation{State Key Laboratory of Autonomous Intelligent Unmanned Systems, Shanghai Research Institute for Intelligent Autonomous Systems, Tongji University, Shanghai, 201203, China}
\affiliation{The Department of Control Science and Engineering, Tongji University, Shanghai 201804, China}

\author{Guofeng Zhang}
\email{guofeng.zhang@polyu.edu.hk}
\affiliation{Department of Applied Mathematics, The Hong Kong Polytechnic University, kowloon, 999077, Hong Kong, China}
\affiliation{Shenzhen Research Institute, The Hong Kong Polytechnic University, Shenzhen, 518000, China}

\author{Re-Bing Wu}
\email{rbwu@tsinghua.edu.cn}
\affiliation{Department of Automation, Tsinghua University, Beijing, 100084, China}

\begin{abstract}
	Optimization plays a central role in variational quantum algorithms, where the objective function typically takes the form $F(\boldsymbol{\theta})= \sum_{m=1}^{M} f_m \left(\mathrm{Tr}[U(\boldsymbol{\theta})\rho_m U^\dagger(\boldsymbol{\theta}) O_m]\right)$, with $U(\boldsymbol{\theta})$ being a parameterized quantum ansatz. Understanding the optimization landscape of such objective functions is crucial for assessing the trainability and performance of these algorithms. For the special case $M=1$, it is known that under certain assumptions, the landscape is free of false traps (FTs), i.e., local optima that are not global. In this work, we investigate optimization landscapes of the general case $M\geq1$ and show that the landscape becomes intrinsically more complex. First, we establish a complete framework for analyzing critical features of the optimization landscape, by deriving necessary and sufficient conditions to identify and classify all critical points under some assumptions, which is also of practical importance in designing efficient algorithms independent of whether these assumptions are satisfied. Then, we show that FTs can still emerge on landscapes for $M>1$, standing in stark contrast to the $M=1$ case and further revealing that parameter sufficiency alone is not enough to guarantee a trap-free landscape. Moreover, we uncover a close connection that the emergence of FTs  is necessarily attributed to the loss of distinguishability among the states and/or operators, and fundamentally, to the loss of compatibility of the spectral ordering governed by different objective terms. Our results 
	provide a deeper understanding of the optimization complexity and practical guidance for both algorithmic and problem-setting designs.

	\end{abstract}
	
	\maketitle

	\section{Introduction}
	
	The rapid development of quantum computing has opened new avenues for solving problems that are intractable for classical computers. In the current noisy intermediate-scale quantum era where fully fault-tolerant quantum computers remain out of reach, variational quantum algorithms (VQAs) have emerged as a promising paradigm for harnessing near-term quantum devices, due to their flexibility, adaptability, and potential for achieving practical quantum advantages in broad areas such as combinatorial optimization~\cite{Grange2024,Li2026}, machine learning~\cite{Biamonte2017,Li2021,Cerezo2022}, and simulation~\cite{Banuls2006,Moueddene2020,Endo2020}. Therefore, they have attracted significant attention~\cite{Amaro2022,Huembeli2021,Wang2024,DiezValle2023,Chiew2024,Peruzzo2014,Farhi2014,Bharti2022,Georgescu2014,Daley2022,Preskill2018,Akshay2020,Abbas2021,Chen2022}.
	
	Typically, a VQA aims to minimize or maximize an objective function of the form~\cite{Beer2020,Cerezo2021a,Thanasilp2023,Urbaneja2025}
	\begin{equation}\label{fun1}
		F(\boldsymbol{\theta})=\sum_{m=1}^{M} f_m\left(\tr[U(\boldsymbol{\theta})\rho_m U^\dagger(\boldsymbol{\theta}) O_m]\right),
	\end{equation} 
	where $U(\boldsymbol{\theta})$ represents quantum ansatz parameterized by a set of classically tunable parameters $\boldsymbol{\theta}=(\theta_1,\dots,\theta_N)^\top$ (e.g., rotation angles of logic gates~\cite{Mitarai2018,Cerezo2021}). Here, $\rho_m$ denotes a quantum state on a Hilbert space of dimension $D$, $f_m$ is some function, and $O_m$ represents an observable, a many-body Hamiltonian, an element of a quantum measurement, or another state. For example, in the single-term scenario ($M=1$), it reduces to the widely studied form $F(\boldsymbol{\theta})=\tr[U(\boldsymbol{\theta})\rho U^\dagger(\boldsymbol{\theta})O]$, which arises in paradigmatic VQAs such as the standard variational quantum eigensolver (VQE)~\cite{Peruzzo2014} and quantum approximate optimization algorithm (QAOA)~\cite{Farhi2014}. For $M>1$, it generally appears in advanced applications including the subspace-search VQE (SSVQE)~\cite{Nakanishi2019}, quantum autoencoders~\cite{Romero2017}, and variational approaches to principle component analysis~\cite{Bravo2020}. Remarkably, similar optimization problems also arises frequently in quantum optimal control, state discrimination and related information-processing tasks~\cite{Bae2015,Ge2021,Ge2022}. 
	
		To tackle the above problem, quantum-classical hybrid optimization routines have been developed, wherein the quantum component directly evaluates the function $F(\boldsymbol{\theta})$ for any given $\boldsymbol{\theta}$ on quantum hardware, while the classical component updates the parameters $\boldsymbol{\theta}$ using an algorithmic optimizer. Since gradient-based optimizers are a common choice due to their empirical efficiency in handling nonconvex objective functions~\cite{Stokes2020}, a natural question then arises: can such optimizers successfully find optimal solutions? The performance of these hybrid routines depends critically on the optimization landscape induced by $F(\boldsymbol{\theta})$, namely the topological structure of the objective function. In practice, the landscape may contain obstacles that prevent optimizers from reaching global optima. A prominent example is the presence of {\it false traps} (FTs), i.e., local but not global optima (as illustrated in Fig.~\ref{fig1}) which can trap optimizers in suboptimal solutions and thereby undermine potential quantum advantages~\cite{Wiedmann2025}. 

	For the special case $M =1$, a key observation from prior work is that the structure of the landscape depends crucially on the number of tunable parameters. When the parameter number $N$ is small, the presence of FTs is a generic property~\cite{Kiani2020,Wierichs2020,Riveradean2021}. As $N$ increase, FTs are expected to gradually disappear~\cite{Lee2021,Larocca2023}. Remarkably, under certain assumptions, the optimization landscape can be free of FTs ~\cite{Rabitz2004,Rabitz2005,Wu2007,Russell2017}. This trap-free property underpins much of the optimism in quantum optimal control~\cite{Ho2006} and suggests that gradient-based optimizers can reliably find global optima in such settings. In contrast, it is much less explored for $M>1$ and especially remains an open question whether FTs can be completely avoided under the same assumptions as those in the single-term case.

\begin{figure*}
	\begin{center}
		\includegraphics[width=0.7\textwidth]{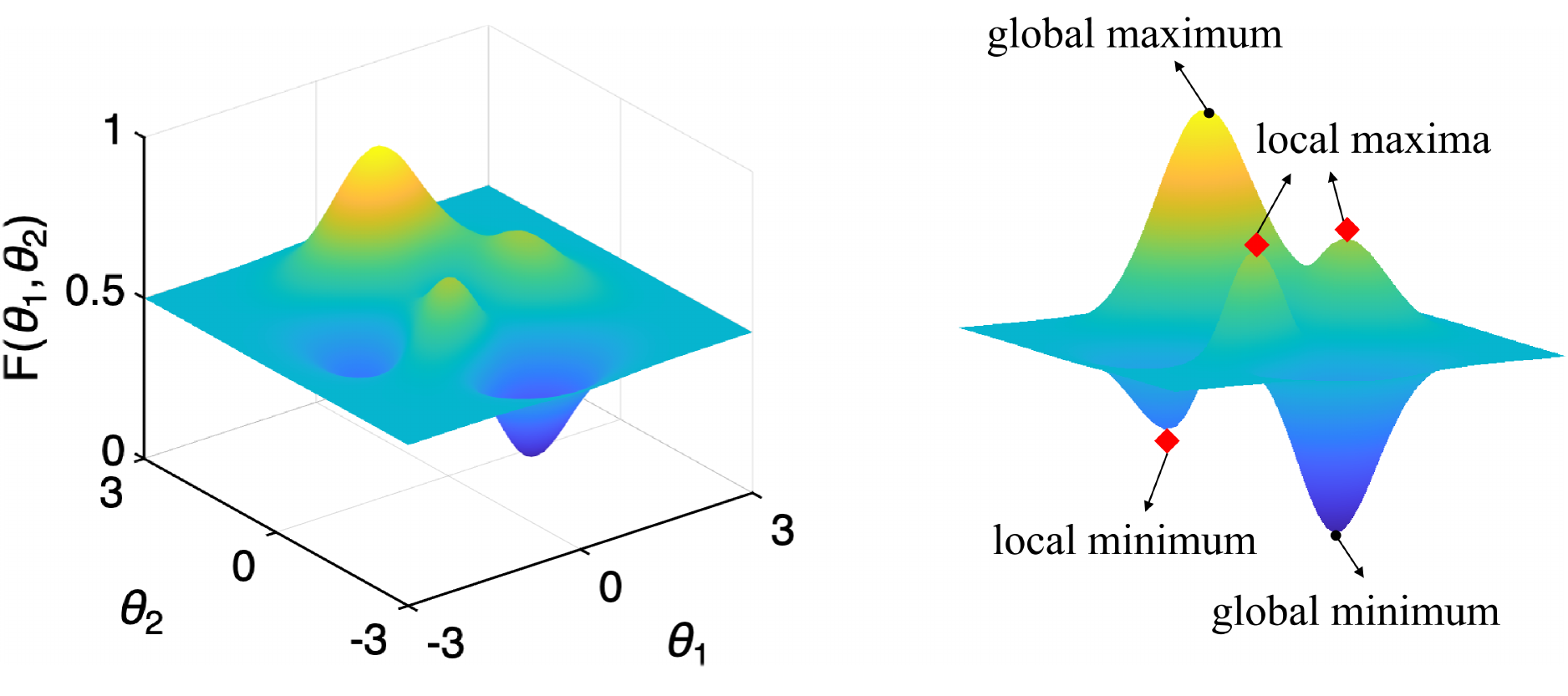}
	\end{center}
	\caption{Illustration of an optimization landscape $F(\boldsymbol{\theta})$ as a function of parameters $\theta = (\theta_1, \theta_2)^\top$. Local maxima (minima), which are not globally optimal with respect to the maximization (minimization) of $F$, are likely to hinder optimizers from reaching globally optimal solutions, thus forming false traps. While the landscapes are known to be trap-free in the single-term case~\cite{Wu2007}, we show that false traps can still emerge in the multi-term case under the same assumptions.} 
	\label{fig1}
\end{figure*}

In this work, we study the optimization landscape of the objective function~(\ref{fun1}), by presenting a complete framework to analyze critical features of optimization landscape under Assumption~\ref{assm1} and~\ref{assm2}. Particularly, we provide necessary and sufficient conditions to identify all critical points and to further classify them as local minima, maxima, or saddles, naturally generalizing Ref.~\cite{Wu2008} from $M=1$ to $M\geq 1$. Importantly, the first- and second-order derivative information obtained within this framework is helpful not only for designing optimization strategies on the Riemannian manifold, but also offers practical guidance for updating parameters in Euclidean geometry~\cite{Wiersema2023}. 
	
Then, an explicit example is constructed in Eq.~(\ref{opt1}) to show the possibility of FTs emerging on the landscape for $M>1$, revealing the fundamental difference between the landscape topologies in regimes $M=1$ and $M>1$. Furthermore, a necessary and sufficient condition is obtained in Theorem~\ref{theo3} to identify FTs among a set of special critical points termed simultaneous, as well as a geometric interpretation. It implies that the presence of FTs cannot be solely attributed to insufficiency of tunable parameters, and also provides a negative answer to the above open question, which is surprisingly different from the $M=1$ case.

Moreover, a deep connection is uncovered between optimization landscape topology and the distinguishability of states and observables in the multi-term objective function. When both sets are assumed perfectly distinguishable, it is proven in Theorem~\ref{theo4} that FTs formed by simultaneous critical points are always absent. Supported with strong numerical evidence, it is further conjectured that the optimization landscape with $M\geq 1$ is trap-free under the perfect distinguishability condition. This reveals that the emergence of FTs is necessarily attributed to the loss of distinguishability. More fundamentally, it is shown that the emergence of FTs is rooted in incompatibility of the spectral ordering generated by different objective terms. These results not only provide a physical  interpretation of landscape complexity but also suggest alternative strategies to mitigate FTs through problem setting designs.

Finally, implications and generalizations of our results are noted, the connection between FTs and frustration in condensed matter systems~\cite{Leon2010} is discussed, and the distinction between FTs and barren plateaus~\cite{McClean2018} is also clarified.

	The remainder of this paper is organized as follows. Section~\ref{sec2} gives a brief introduction to the optimization landscape of the objective function~(\ref{fun1}) and establishes the landscape equivalence between the objective~(\ref{fun1}) and~(\ref{Uloss}), under Assumptions~\ref{assm1} and~\ref{assm2}. Section~\ref{sec3} presents a complete framework for analyzing optimization landscapes and revisits the single-term case. In Sec.~\ref{sec4}, the established framework is further applied to study the general optimization landscape, with an emphasis on simultaneous critical points, and the existence of FTs is explored. Finally, discussions and concluding remarks are provided in Sec.~\ref{sec5} and Sec.~\ref{sec6}, respectively.

\section{Preliminaries}\label{sec2}

\subsection{The optimization landscape}

Consider an optimization problem with the objective function $F(\boldsymbol{\theta})$ given in~(\ref{fun1}). Our aim is to use some algorithmic optimizer to find at least one global optimum denoted by $\boldsymbol{\theta}^*$, such that $F(\boldsymbol{\theta}^*)$ reaches its maximal or minimal value. Indeed, whether the used optimizer can efficiently and successfully find one $\boldsymbol{\theta}^*$ essentially depends on the optimization landscape $\mathcal{L}_{\boldsymbol{\theta}}(F)$ generated by $F(\boldsymbol{\theta})$. 

As illustrated in Fig.~\ref{fig1}, any optimum necessarily corresponds to a critical point $\boldsymbol{\theta}^\prime$ on $\mathcal{L}_{\boldsymbol{\theta}}(F)$, satisfying
\begin{equation}\nonumber
	\frac{\partial F}{\partial \boldsymbol{\theta}^\prime}:=(\partial F/\partial \theta^\prime_1,\cdots,\partial F/\partial \theta^\prime_N)^\top=\boldsymbol{0},
\end{equation}
 and each critical point is further classified as 
\begin{itemize}
	\item [$\bullet$] \textit{global optimum}: $F(\boldsymbol{\theta}^\prime)\geq F(\boldsymbol{\theta})$ (global maximum) or $F(\boldsymbol{\theta}^\prime)\leq F(\boldsymbol{\theta})$ (global minimum) for all $\boldsymbol{\theta}$.
	\item[$\bullet$] \textit{local optimum}: $ F(\boldsymbol{\theta}^\prime)\geq F(\boldsymbol{\theta})$ (local maximum) or $F(\boldsymbol{\theta}^\prime)\leq F(\boldsymbol{\theta})$ (local minimum) for all $\boldsymbol{\theta}$ within a neighborhood of $\boldsymbol{\theta}^\prime$.
	\item[$\bullet$] \textit{saddle point}: $\boldsymbol{\theta}^\prime$ is neither a local maximum nor a local minimum.
\end{itemize}
The above classification task can be accomplished via the associated Hessian matrix defined as
\begin{equation} \nonumber
	H_{\boldsymbol{\theta}}:=(H)_{N\times N},~~H_{nn^\prime}=\frac{\partial^2 F}{\partial \theta_n\partial \theta_{n^\prime}}, ~~ n, n^\prime=1,\dots, N.
\end{equation}
For a critical point $\boldsymbol{\theta}^\prime$, if its Hessian matrix is negative (positive) semidefinite, then it corresponds to a local maximum (minimum); Otherwise, it is a saddle. Generally, determining global optimum is challenging, as it requires enumerating all possible critical points.

It is also illustrated in Fig.~\ref{fig1} that there possibly exist critical points on the landscape that are local but not global optima, likely to prevent algorithmic optimizers (e.g., gradient-based ones) from successfully finding optimal solutions to the optimization problem and hence forming FTs. 

\begin{Definition}
	False traps are critical points on the optimization landscape that are locally but not globally optimal. Formally, a false trap is called a critical point $\boldsymbol{\theta}^\prime$ for maximizing (minimizing) $F(\boldsymbol{\theta})$, if $\partial F/\partial \boldsymbol{\theta}^\prime=\boldsymbol{0}$, $H_{\boldsymbol{\theta}^\prime}\leq \boldsymbol{0}$ ($H_{\boldsymbol{\theta}^\prime}\geq \boldsymbol{0}$), and $F(\boldsymbol{\theta}^\prime)\neq F(\boldsymbol{\theta}^*)$.
\end{Definition}

The presence of FTs poses a severe bottleneck for quantum-classical optimization, which is also detrimental for witnessing quantum advantage of the underlying information processing task. Therefore, a detailed analysis of FTs is demanding and proper tools are needed to mitigate this phenomenon.

\subsection{The landscape equivalence}

Note from Eq.~(\ref{fun1}) that the objective function $F(\boldsymbol{\theta})$ is composed of two sequential mappings
\begin{equation}\label{map}
	\boldsymbol{\theta} \rightarrow U \rightarrow F.
\end{equation}
The first maps real parameters $\boldsymbol{\theta}\in\mathbb{R}^N$ to quantum unitary ansatz $U\in \mathcal{U}(D)$, which is then mapped to a real number $F\in \mathbb{R}$ by the second. Evidently, it is more straightforward to optimize the function $F$ over ansatz $U$ than parameters $\boldsymbol{\theta}$, and also it is easier to study the optimization landscape $\mathcal{L}_U(F)$ induced by 
\begin{equation}\label{Uloss}
	F(U):=\sum_{m=1}^{M} f_m(\tr[U\rho_m U^\dagger O_m])
	\end{equation}
 rather than $\mathcal{L}_{\boldsymbol{\theta}}(F)$. However, we cannot directly replace $\boldsymbol{\theta}$ with $U$, because $\mathcal{L}_{U}(F)$ is generally not equivalent to $\mathcal{L}_{\boldsymbol{\theta}}(F)$. 

We introduce two assumptions on the first mapping in Eq.~(\ref{map}) to establish the landscape equivalence between $F(\boldsymbol{\theta})$ and $F(U)$. In particular, the first is given as:

\begin{Assumption}\label{assm1}
Any operator $U\in\mathcal{U}(D)$ can be realized by some admissible $\boldsymbol{\theta}$.
\end{Assumption}

It indicates the function range of the ansatz domain is equivalent to that of the parameter domain, i.e, $F(\mathcal{U}(D))= F(\mathbb{R}^N)$, thus establishing the zero-th equivalence between $\mathcal{L}_{U}(F)$ and $\mathcal{L}_{\boldsymbol{\theta}}(F)$. Additionally, the second is:

\begin{Assumption}\label{assm2}
	The Jacobin matrix $\partial U/\partial \boldsymbol{\theta}$ is nonsingular for any possible $\boldsymbol{\theta}$.
\end{Assumption}

This assumption ensures that $\boldsymbol{\theta}^\prime$ is a critical point of $F(\boldsymbol{\theta})$ if and only if $U(\boldsymbol{\theta}^\prime)$ is a critical point of $F(U)$, and furthermore, the type of $\boldsymbol{\theta}^\prime$ (local maximum, minimum, or saddle) is identical to that of $U(\boldsymbol{\theta}^\prime)$~\cite{Wu2007,Ge2022}. Consequently, the landscape equivalence, up to the second order, is established for $F(\boldsymbol{\theta})$ and $F(U)$, and correspondingly, Eq.~(\ref{map}) is strengthened to
\begin{align}
	\boldsymbol{\theta} \rightleftharpoons U \rightarrow F,
\end{align}
with the second-order equivalence $\rightleftharpoons$, allowing us to focus on the optimization landscape $\mathcal{L}_{U}(F)$ for $F(\boldsymbol{\theta})$.

It is remarked that the above assumptions are ${\it strong}$, in the sense that Assumption~\ref{assm1} needs sufficiently many parameters and Assumption~\ref{assm2} requires any local direction around parameters be attainable, which may not always be satisfied~\cite{Ge2021,Anschuetz2022}. It is also pointed out that they have been used in quantum optimal control theory to study control landscapes~\cite{Wu2008,Wu2007,Hsieh2010,Rabitz2004,Rabitz2005}.

\section{Identification and Classification of Critical Points}\label{sec3}

In this section, we first establish a complete framework to study optimization landscapes of $F(U)$ in~(\ref{Uloss}) and $F(\boldsymbol{\theta})$ in~(\ref{fun1}) under Assumptions~\ref{assm1} and~\ref{assm2}, by obtaining necessary and sufficient conditions to identify all critical points and to determine their types. Then, we illustrate how it works for the $M=1$ case, by deriving analytical forms of critical points and proving the absence of FTs. It is worth noting that this framework is of practical importance in designing algorithmic optimizers, independent of whether Assumptions~\ref{assm1} and~\ref{assm2} are satisfied or not. 

\subsection{The framework for optimization landscapes}

Given any optimization problem with the objective function~(\ref{fun1}) or~(\ref{Uloss}), we are able to obtain

\begin{Theorem} \label{theo1}
$U$ is a critical point of the objective function $F(U)$ in~(\ref{Uloss}),  if and only if 
	\begin{equation}\label{condition}
		\sum_{{m}=1}^{M}\omega_m(U) [U\,\rho_{m}\,U^{\dagger}, O_{m}]=0,
	\end{equation}	
	where $\omega_m(U)=f_m^\prime(l_m(U))$ describes the derivative of $f_m$ with respect to $l_m(U):=\tr[U\rho_m U^\dagger O_m]$ and $[\cdot,\cdot]$ the commutator. Moreover, $\boldsymbol{\theta}^\prime$ is a critical point of $F(\boldsymbol{\theta})$ in~(\ref{fun1}), if and only if the above critical-point condition holds with $U=U(\boldsymbol{\theta}^\prime)$, under Assumptions~\ref{assm1} and~\ref{assm2}. 
\end{Theorem}

The proof is given as follows. First, the neighborhood of an arbitrary unitary ansatz $U$ is parametrized as 
\begin{align}\label{parameterisation}
U(s,A)=e^{isA}U
\end{align}
with $s\in \mathbb{R}$ and $A^\dagger =A$, and the first-order derivative of $F(U(s, A))$, with respect to $s$, at the point $s=0$ is 
\begin{equation}\label{direction}
	\left.\frac{d}{d\,s} F\left(U(s,A)\right)\right|_{s=0}= \tr\left[ iA\sum_{m=1}^{M}\omega_m(U)[U\rho_{m}U^{\dagger},O_{m}]\right].
\end{equation}
Then, $U$ is a critical point of $F(U)$ if and only if the above derivative is zero for any Hermitian $A$.  Thus, the critical-point condition~(\ref{condition}) is derived as desired. Finally, recall from Assumptions~\ref{assm1} and~\ref{assm2} that $\boldsymbol{\theta}^\prime$ is a critical point of $F(\boldsymbol{\theta})$ if and only if $U(\boldsymbol{\theta}^\prime)$ is a critical point of $F(U)$, and consequently, the condition~(\ref{condition}) becomes
\begin{align}
		\sum_{{m}=1}^{M}\omega_m(U)[U(\boldsymbol{\theta}^\prime)\,\rho_{m}\,U^{\dagger}(\boldsymbol{\theta}^\prime), O_{m}]=0
\end{align}
for any critical point $\boldsymbol{\theta}^\prime$. This completes the proof.

Theorem~\ref{theo1} provides a necessary and sufficient condition to identify critical points on the landscapes. To further classify them, we examine the second derivative. Utilizing again the parametrisation form~(\ref{parameterisation}) leads to the second-order derivative of $F(U(s, A))$ at point $s=0$
\begin{eqnarray}
	h_{U}(A)&=&	\left.\frac{d^2}{d^2\, s} F\left(U(s,A)\right)\right|_{s=0} \nonumber \\
	&=&\sum_{m=1}^{M}\omega_m(U)\tr[U\rho_m U^\dagger(AO_m A-\frac{1}{2}\{O_m,A^2\})] \nonumber\\
	&&-f^{\prime\prime}_m(l_m(U))\tr^2\left[A[U\rho_mU^\dagger,O_m]\right]\label{gHess}\\
	&=&\sum_{m=1}^{M} \omega_m(U)\tr[U\rho_m U^\dagger(AO_m A-O_m A^2)] \label{Hess}
\end{eqnarray}
with the anti-commutator $\{\cdot,\cdot \}$. Here, the third equality follows from the critical-point condition~(\ref{condition}) and $\tr[XY]=\tr[YX]$ for matrices $X, Y$, while the second holds independent of whether the point $U$ is critical or not.

Denote further by ${\rm vec}(X)$ the vectorization operation on matrix $X$ and by ${\rm vech}(X)$ half-vectorization on $X$ with certain symmetry (see Appendix~\ref{proof1} for more details). It is proven in Appendix~\ref{proof2} that by decomposing matrices into real and imaginary components, i.e., $T=T_{\rm re}+iT_{\rm im}$ and $A=A_{\rm re}+iA_{\rm im}$, we can rewrite the above second-order derivative into a quadratic form of
\begin{align}
		{\rm vec}(A)^\dagger\, T\, {\rm vec}(A)
&=	\begin{bmatrix}{\rm vech}(A_{\rm re})\\{\rm vech}(A_{\rm im})
\end{bmatrix}^\top \begin{bmatrix}
	D_{\rm sy} & \boldsymbol{0}\\
	\boldsymbol{0} & D_{\rm ay}
\end{bmatrix}^\top
 \nonumber \\
&\begin{bmatrix}
	T_{\rm re} & -T_{\rm im}\\
	T_{\rm im} & T_{\rm re}
\end{bmatrix} \begin{bmatrix}
	D_{\rm sy} & \boldsymbol{0}\\
	\boldsymbol{0} & D_{\rm ay}
\end{bmatrix}\begin{bmatrix}{\rm vech}(A_{\rm re})\\{\rm vech}(A_{\rm im})
\end{bmatrix} \nonumber \\
	&\equiv \boldsymbol{v}^\top H \boldsymbol{v} \label{Hess2}
\end{align}
with $T=\sum_{m} \omega_m(U)[O_m^\top\otimes U\rho_m U^\dagger-\mathbb{I}\otimes U\rho_m U^\dagger O_m]$, duplication matrix $D_{\rm sy}$ for the symmetric $A_{\rm re}$, and $D_{\rm ay}$ for the antisymmetric $A_{\rm im}$. All elements of vector $\boldsymbol{v}=\begin{bmatrix}{\rm vech}(A_{\rm re})\\{\rm vech}(A_{\rm im})\end{bmatrix}\in\mathbb{R}^{D^2}$ can be arbitrarily valued, so the matrix
\begin{align}
	H=\begin{bmatrix}
		D_{\rm sy} & \boldsymbol{0}\\
		\boldsymbol{0} & D_{\rm ay}
	\end{bmatrix}^\top\begin{bmatrix}
		T_{\rm re} & -T_{\rm im}\\
		T_{\rm im} & T_{\rm re}
	\end{bmatrix}\begin{bmatrix}
		D_{\rm sy} & \boldsymbol{0}\\
		\boldsymbol{0} & D_{\rm ay}
	\end{bmatrix} \label{Hessmatrix}
\end{align}
can be considered as the corresponding Hessian matrix, by noting from Eq.~(\ref{Hess}) that any critical point $U$ is a local maximum (minimum) if and only if $h_{U}(A)$ is nonpositive (nonnegative) for any $A$, and further from Eq.~(\ref{Hess2}) that $h_{U}(A)$ is always nonpositive (nonnegative) if and only if matrix $H$ in Eq.~(\ref{Hessmatrix}) is negative (positive) semidefinite. This immediately yields

\begin{Theorem}\label{theo2}
	Any critical point $U$ of the objective function $F(U)$ in~(\ref{Uloss}) is locally maximal (minimal) if and only if its Hessian matrix in Eq.~(\ref{Hessmatrix}) is negative (positive) semidefinite. Moreover, a critical point $\boldsymbol{\theta}^\prime$ is locally maximal (minimal) for $F(\boldsymbol{\theta})$ in~(\ref{fun1}) if and only if the Hessian matrix generated by $U(\boldsymbol{\theta}^\prime)$ is negative (positive) semidefinite, under Assumptions~\ref{assm1} and \ref{assm2}.
	\end{Theorem}

Based on these two theorems, a complete framework is built up to reveal critical features of every optimization landscape, in the sense that Theorem~\ref{theo1} provides a criterion to identify critical points and Theorem~\ref{theo2} further provides a feasible way to classify them. Moreover, it also works for non-critical points, and importantly, provides practical guidance to design gradient-based algorithms for the optimization of $F(U)$ and $F(\boldsymbol{\theta})$. 

Particularly, the critical-point condition~(\ref{condition}) must be violated by all non-critical points, and thus the first-order derivative in~(\ref{direction}) can be nonzero for some $A$. If matrix $A$ in $U(s,A)$ is chosen as
\begin{align}\label{direction3}
	A=\mp  i\sum_{m=1}^M \omega_m(U)[U\rho_m U^\dagger, O_m],
\end{align}
then $\left.d F\left(U(s,A)\right)/d s\right|_{s=0}=\pm \tr\left[A^\dagger A \right]\neq 0$, implying $F(U)$ increases or decreases along the direction $A$ at the point $U$. As a consequence, a gradient update rule for optimizing $F(U)$ is obtained as
\begin{equation}\label{strategy}
	U\rightarrow e^{i\eta A}U
\end{equation}
with a learning rate $\eta$. This gradient generation process can be regarded as Riemannian gradient flow on manifold $\mathcal{U}(D)$. Furthermore, the Hessian matrix of non-critical points can be similarly constructed from Eq.~(\ref{gHess}), which is of practical importance in adjusting learning rates and in avoiding saddles by choosing $A$ such that $h_{U}(A)>0$ for maximizing $F(U)$ or $h_U(A)<0$ for minimizing $F(U)$. 

When it comes to optimizing $F(\boldsymbol{\theta})$, it is a good choice to update parameters $\boldsymbol{\theta}$ in the Euclidean geometry along the direction that approximates the Riemannian gradient flow~(\ref{strategy})~\cite{Wiersema2023} and to adjust learning rates based on the corresponding $h_{U(\boldsymbol{\theta})}$, both of which can always be realized independent of whether Assumptions~\ref{assm1} and/or~\ref{assm2} are satisfied or not.

\subsection{The $M=1$ case revisited}\label{M1}

Consider the $M=1$ case where function $f$ is linear, i.e., its first-order derivative is constant, which has wide applications in variational quantum eigensolver and quantum approximate optimization algorithm. It follows first from Theorem~\ref{theo1} that the critical-point condition~(\ref{condition}) simplifies to 
\begin{align}
[U\,\rho\,U^{\dagger}, O]=0,
\end{align}
for $F(U)$ with a single state $\rho$ and operator $O$. Denote by $P$ the unitary transformation that diagonalizes $\rho$ into $\hat{\rho}=P\rho P^\dagger={\rm diag}\{\underbrace{\lambda_1,\cdots,\lambda_1}_{p_1},\cdots,\underbrace{\lambda_r,\cdots,\lambda_r}_{p_r} \}$ with eigenvalues $\lambda_1>\cdots>\lambda_r$ and multiplicities $p_1,\cdots,p_r$, and by $ Q$ the unitary transformation such that $\hat{O}=Q\,O Q^\dagger={\rm diag}\{\underbrace{o_1,\cdots,o_1}_{q_1},\cdots,\underbrace{o_s,\cdots,o_s}_{q_s} \}$ with eigenvalues $o_1>\cdots>o_s$ and multiplicities $q_1,\cdots,q_s$. It is shown in~\cite{Wu2007} that all critical points are given in a set
	\begin{align} \label{solution1}
		\mathcal{S}=\{Q^\dagger U^\dagger_q\pi U_p P: U_q\in \mathcal{U}(\boldsymbol{q}), U_p\in \mathcal{U}(\boldsymbol{p}), \pi\in \mathcal{P}\},
	\end{align}
where $\mathcal{U}(\boldsymbol{p})$ represents the product group $\mathcal{U}(p_1)\times\cdots\times\mathcal{U}(p_r)$ with $\mathcal{U}(p_i)$ the $p_i$-dimensional unitary group for $i=1,\dots,r$, $\mathcal{U}(\boldsymbol{q})=\mathcal{U}(q_1)\times\cdots\times\mathcal{U}(q_s)$ similarly defined, and $\mathcal{P}$ the permutation group. 

Then, substituting any critical point $U^\prime=Q^\dagger U^\dagger_q\pi U_p P$ into $F$ yields
\begin{align}
	F(U^\prime)=\tr[\hat{O}\pi\hat{\rho}\pi^\dagger]=\sum_{k=1}^D \lambda_{(k)}\,o_k,
\end{align}
where $\lambda_{(k)}, o_k$ refer to the $k$-th diagonal entry of $\pi \hat{\rho} \pi^\dagger, \hat{O}$, respectively, and into the Hessian matrix~(\ref{Hessmatrix}) yields that its diagonal entries coincide with its eigenvalues
\begin{equation}\label{eigs}
	h_{kk'}=-(\lambda_{(k)}-\lambda_{(k')})(o_k-o_{k'}), ~~1\leq k\leq k'\leq D.
\end{equation}
The complete derivation is deferred to Appendix~\ref{appendixA.3}.

Following finally from Theorem~\ref{theo2} and Eq.~(\ref{eigs}) yields that $U^\prime$ is a local maximum (minimum) if and only if all $h_{kk'}$s are nonpositive (nonnegative). Note that $o_k$s in Eq.~(\ref{eigs}) are in descending order, i.e., $o_k\geq o_{k'}$ for $k<k'$, so $\lambda_{(k)}$s must obey the descending (ascending) order of $\lambda_{(k)}\geq \lambda_{(k')}$  ($\lambda_{(k)}\leq \lambda_{(k')}$) to ensure nonpositive (nonnegative) $h_{kk^\prime}$s. Thus, all local maxima (minima) require $\lambda_{(k)}$s in the descending (ascending) order, leading to an identical  maximal (minimal) landscape value $F(U^\prime)$. This recovers a remarkable result obtained in quantum optimal control theory~\cite{Rabitz2004,Rabitz2006,Ho2006,Wu2007} that the landscape is devoid of FTs for $M=1$.

It is remarked that a structural property of critical points on the single-term landscape is revealed via Eq.~(\ref{eigs}). Indeed, each critical point is associated with a spectral ordering between the state and the observable, specified by a permutation $\pi$. Particularly, all local maxima (or minima) admit a common spectral ordering, explaining why FTs are avoided. In the following, it serves as a useful benchmark for revealing the richer behavior of general optimization landscapes.

\section{Landscapes of Multi-Term Objectives}\label{sec4}

We continue to apply the above framework study the optimization landscape generated by the multi-term objective
\begin{equation}\label{linear}
	F(U)=\sum_{m=1}^{M}\omega_m
	\tr(U\rho_mU^\dagger O_m),
\end{equation}
where $f_m$s are linear, i.e., $f_m(x)=\omega_mx$. As a general nonlinear function $f$ could change the critical properties of optimization landscapes and even lead to the emergence of FTs in the single-term case, we mainly focus on the linear setting to reduce the possibility that the landscape complexity originates purely from nonlinear transformations of expectation values. Furthermore, it enables a direct case comparison between $M=1$ and $M>1$. 

Generally, it is challenging to derive explicit forms of critical points on the optimization landscape of $F(U)$. Therefore, we study a class of critical points, termed simultaneous critical points.
\begin{Definition}
	A simultaneous critical point is a critical point of
	$F(U)$ satisfying
	\begin{equation}\label{simcritical}
		[U\rho_mU^\dagger,O_m]=0, \qquad \forall\,m=1,\dots,M.
	\end{equation}
\end{Definition}
It follows from the above subsection that simultaneous critical points are in the following set
\begin{align}\label{generalsolution}
	\mathcal{S}=\cap_m^M\mathcal{S}_m,
\end{align}
where each $\mathcal{S}_m$ is given as Eq.~(\ref{solution1}) with $\rho_m$ and $O_m$. Evidently, it might be empty.

In this section, we first give one sufficient condition to ensure the nonemptiness of $\mathcal{S}$ and then obtain an explicit form of simultaneous critical points. Surprisingly, we find that FTs can emerge on optimization landscapes and further derive one necessary and sufficient condition to identify FTs. Finally, the close connection between the presence of FTs and indistinguishability is explored.

\subsection{Simultaneous critical points}

 We provide one sufficient condition to guarantee the existence of simultaneous critical points in the objective function~(\ref{linear}).

\begin{Lemma}\label{lemma1}
	If the states $\rho_1,\dots,\rho_M$ are pairwise commutative and the observables $O_1,\dots,O_M$ are pairwise commutative in the objective function~(\ref{linear}), then simultaneous critical points exist, or equivalently, the set $\mathcal{S}$  given as Eq.~(\ref{generalsolution}) is nonempty.
\end{Lemma}

Notably, all pairwise-commutative states can be simultaneously diagonalized into 
\begin{equation}\label{mainform}
	\hat{\sigma}_m=\omega_m\hat{P} \rho_m \hat{P}^\dagger ={\rm blkdiag}\{ \underbrace{\hat{\Lambda}^m_1}_{d_1},\cdots,  \underbrace{\hat{\Lambda}^m_{i=m}}_{d_m},\cdots,\underbrace{\hat{\Lambda}^m_{M}}_{d_M} \}
\end{equation}
by one single unitary transformation $\hat{P}$. $\hat{\Lambda}_{i}^{m}$ represents a $d_i$-dimensional diagonal matrix, and especially, the block $\hat{\Lambda}^m_{m}$ is element-wise dominant of $\hat{\sigma}_m$, in the sense that the $k$-th diagonal element $\hat{\lambda}^m_k$ of $\hat{\sigma}_m$ is dominant if it is larger than those of the rest $\hat{\sigma}_{m'}$, i.e., $\hat{\lambda}_{k}^m\geq\hat{\lambda}_{k}^{m'}$ for all $m'\neq m$ and strictly holds for $m^\prime<m$.  As the diagonal $\hat{\sigma}_m$ can be faithfully described by vector $\boldsymbol{r}_m=(\hat{\lambda}^m_1,\cdots,\hat{\lambda}^m_D)^\top \in \mathbb{R}^D$, all $\hat{\sigma}_m$s give rise to a compact matrix
\begin{align*}
	\hat{G}_\rho=\begin{bmatrix} \boldsymbol{r}^\top_1\\ \boldsymbol{r}^\top_2 \\\cdot \\\cdot \\\cdot \\ \boldsymbol{r}^\top_M \end{bmatrix}=[\underbrace{\hat{\boldsymbol{\lambda}}_1,\cdots,\hat{\boldsymbol{\lambda}}_1}_{p_1},\cdots,\underbrace{\hat{\boldsymbol{\lambda}}_r,\cdots,\hat{\boldsymbol{\lambda}}_r}_{p_r}],
\end{align*}
where $\hat{\boldsymbol{\lambda}}_1,\cdots,\hat{\boldsymbol{\lambda}}_r$ are distinct vectors with multiplicities $p_1,\cdots, p_r$. The derivation is detailed in Appendix~\ref{appendixB}. 

Likewise, all pairwise-commutative observables are simultaneously diagonalized into 
\begin{equation}\label{mainoform}
	\hat{O}_m=\hat{Q}O_m\hat{Q}^\dagger ={\rm blkdiag}\{ \underbrace{\hat{\Omega}_1^m}_{d^\prime_1},\cdots,\underbrace{\hat{\Omega}_m^m}_{d^\prime_m},\cdots,\underbrace{\hat{\Omega}_M^m}_{d^\prime_M}  \},
\end{equation} 
under unitary transformation $\hat{Q}$, with $d^\prime_j$-dimensional diagonal blocks $\hat{\Omega}^m_j$ and the dominant $\hat{\Omega}_m^m$ for $\hat{O}_m$, and correspondingly, another compact matrix is generated as
\begin{equation*}
	\hat{G}_o=[\underbrace{\hat{\boldsymbol{o}}_1,\cdots,\hat{\boldsymbol{o}}_1}_{q_1},\cdots,\underbrace{\hat{\boldsymbol{o}}_s,\cdots,\hat{\boldsymbol{o}}_s}_{q_s}].
\end{equation*}
where $\hat{\boldsymbol{o}}_1,\cdots,\hat{\boldsymbol{o}}_r$ are distinct vectors with multiplicities $q_1,\cdots, q_s$. 

The set of simultaneous critical points is then explicitly given as
\begin{align} \label{solution2main}
	\mathcal{S}=\{\hat{Q}^\dagger U^\dagger_q\pi U_p \hat{P}: U_q\in \mathcal{U}(\boldsymbol{q}), U_p\in \mathcal{U}(\boldsymbol{p}), \pi\in \mathcal{P}\}
\end{align}
with product groups $\mathcal{U}(\boldsymbol{q})=\mathcal{U}(q_1)\times\cdots\times\mathcal{U}(q_s)$ and $\mathcal{U}(\boldsymbol{p})=\mathcal{U}(p_1)\times\cdots\times\mathcal{U}(p_r)$. It admits the same form as Eq.~(\ref{solution1}) derived for the case of $M=1$, but differs in product unitary transformations. Substituting any simultaneous critical point $U^\prime=\hat{Q}^\dagger U^\dagger_q\pi U_p \hat{P}$ into $F$ yields
\begin{equation}
	F(U^\prime)=\sum_{m=1}^M \tr(\pi\hat{\sigma}_m\pi^\dagger\hat{O}_m)=\sum_{m=1}^M\sum_{k=1}^D\hat{\lambda}_{(k)}^m\hat{o}^m_k,
\end{equation}
and into the Hessian form~(\ref{Hess}) yields
\begin{align}
	&h_{U^\prime}(A)=\sum_{m=1}^M \tr[\pi\hat{\sigma}_m\pi^\dagger(\hat{A}\hat{O}_m\hat{A}-\hat{O}_m\hat{A}^2)]\nonumber\\
	=&\sum_{1\leq k\leq k'\leq D}\sum_{m=1}^M -(\hat{\lambda}_{(k)}^m-\hat{\lambda}_{(k')}^m)(\hat{o}_{k}^m-\hat{o}_{k'}^m)(\hat{x}_{kk'}^2+\hat{y}_{kk'}^2) \nonumber \\
	\triangleq& \sum_{1\leq k \leq k'\leq D} h_{kk'}(\hat{x}_{kk'}^2+\hat{y}_{kk'}^2), \label{Hessform}
\end{align}
where $\hat{\lambda}_{(k)}^m, \hat{o}_k^m$ are the $k$-th diagonal entry of $\pi\hat{\sigma}_m\pi^\dagger, \hat{O}_m$, and $\hat{x}_{ij}, \hat{y}_{ij}$ the real and imaginary parts of $(i,j)$ element of $\hat{A}=U_q\,\hat{Q} A\,\hat{Q}^\dagger U_q^\dagger$. Furthermore, all eigenvalues of the Hessian matrix~(\ref{Hessmatrix}) are derived as
\begin{align} \label{geigens}
	h_{kk'}=\sum_{m=1}^M-(\hat{\lambda}_{(k)}^m-\hat{\lambda}_{(k')}^m)(\hat{o}_{k}^m-\hat{o}_{k'}^m), 1\leq k\leq k^{\prime}\leq D,
\end{align}
which generalizes Eq.~(\ref{eigs}) from $M=1$ to $M\geq 1$. It implies that $h_{U}^\prime$ is nonnegative (nonpositive) for any $A$ if and only if all $h_{kk'}$s are nonnegative (nonpositive), and $U^\prime$ is a local minimum (maximum) if and only if all $h_{kk'}$s must be nonnegative (nonpositive).

It is finally remarked that each objective term specifies a preferred spectral ordering between $\hat{\sigma}_m$ and $\hat{O}_m$. In contrast to the single-term case, however, the preferred spectral orderings associated with different terms may not be mutually compatible. More specifically, different terms may favor distinct eigenvalue alignments that cannot be simultaneously realized by a common unitary transformation. We refer to this situation as spectral-ordering incompatibility. As will be shown later, such incompatibility plays a central role in the emergence of the richer landscape structures observed in multi-term objectives and ultimately gives rise to FTs.

\subsection{Discovery of false traps}\label{example1}

Consider an example in which $\omega_m=1/3$, and
\begin{align} \label{opt1}
	\rho_1&=\text{diag}\{ 0.4,0.35,0.15,0.1 \}, ~~~~O_1={\rm diag}\{1,0,0,0\},\nonumber\\
	\rho_2&=\text{diag}\{0.35,0.45,0.2,0\}, ~~~~~~O_2={\rm diag}\{0,1,1,0\},\nonumber\\
	\rho_3&=\text{diag}\{0.29,0.41,0.18,0.12\}, ~O_3={\rm diag}\{0,0,0,1\}.\nonumber 		 
\end{align}
Denote by $\pi(i,j)$ the elementary matrix permuting the $i$-th and $j$-th rows of any matrix. Using Theorems~\ref{theo1} and~\ref{theo2}, one can verify that both of the two simultaneous critical points, $U_1=\mathbb{I}$ and $U_2=\pi(3,4)\pi(1,4)$, are local maxima, but lead to different landscape values, $F(U_1)=0.39$ and $F(U_2)=0.36$. It immediately follows that $U_2$ emerges as a FT on the landscape, which is also confirmed by numerical experiments as shown in Fig.~\ref{fig2}. Surprisingly, this is different from the $M=1$ case in which there is no FT as proven in the above subsection.

\begin{figure}
	\begin{center}
		\includegraphics[width= \columnwidth]{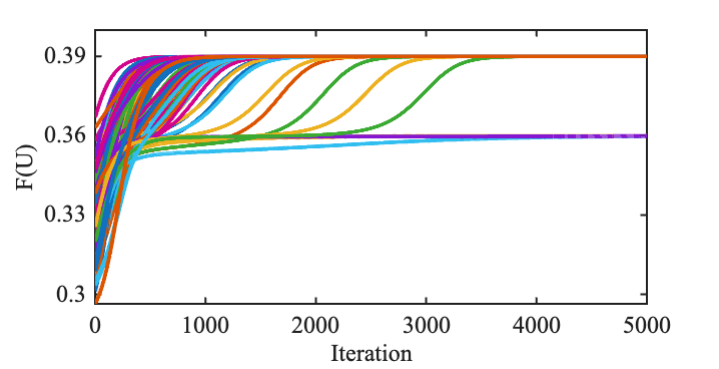}
	\end{center}
	\caption{Numerical simulation on the example~(\ref{opt1}). A gradient-ascend optimizer is utilized to find its maximal value, and $100$ experiments are implemented. It is found that the search for global optimum can be trapped at the critical point $U_2$, leading to $F(U_2)=0.36$ and thus forming a FT.}
	\label{fig2} 
\end{figure}

As noted in Sec.~\ref{M1}, all local optima in the single-term case admit a consistent spectral ordering between the state and the observable, ensuring the absence of FTs. In the above example, the preferred spectral orderings of different objective term are specified by permutations $\pi_1=\mathbb{I}$, $\pi_2=\pi(1,3)$, and $\pi_3=\pi(2,4)$, respectively. Since these permutations correspond to distinct eigenvalue alignments, no single unitary transformation can simultaneously realize all of them. This incompatibility of spectral-ordering preferences induces a competition among the objective terms and leads to emergence of FTs.

 \subsection{Necessary and sufficient condition for false traps}
 
The objective function in~(\ref{linear}) can be rewritten as
 \begin{align}
 F(U)=\sum_{m=1}^M \bar{\omega}_m\tr[U\rho_m U^\dagger \bar{O}_m]+\alpha,
 \end{align}
where $\bar{O}_m=[O_m-\lambda_{\min}(O_m)\mathbb{I}_D]/[\lambda_{\max}(O_m)-\lambda_{\min}(O_m)]$ with the smallest and largest eigenvalue $\lambda_{\min/\max}$ is the rescaled operator, $\bar{\omega}_m=\omega_m[\lambda_{\max}(O_m)-\lambda_{\min}(O_m)]$ the rescaled weight, and $\alpha=\sum_{m=1}^M \omega_m \lambda_{\min}(O_m)$ a constant parameter. Without loss of generality, each $O_m$ can be assumed to satisfy $\boldsymbol{0}\leq O_m\leq \mathbb{I}_D$, and further, $O_m$s can be assumed to form a positive operator-valued measure (POVM), i.e., $\sum_{m=1}^MO_m=\mathbb{I}_D$, by first rescaling all operators to satisfy $\sum^M_mO_m\leq \mathbb{I}_D$ and then adding the constant term $\tr[U(\mathbb{I}/D)U^\dagger O_{M+1}]$ with a new operator $O_{M+1}:=\mathbb{I}_D-\sum^M_m O_m$ into $F(U)$.

 In the case of $M=2$, if operators satisfy $O_1+O_2=\mathbb{I}$, then $F(U)=\tr(U\varrho U^\dagger O_1)+\omega_2$ with $\varrho=\omega_1\rho_1-\omega_2\rho_2$, which reduces to the $M=1$ case. It immediately yields
\begin{Corollary}\label{corollary}
	The optimization landscape of the objective function $F(U)$ in~(\ref{linear}) with $M=2$ is devoid of false traps, if the operator set $\{O_1, O_2\}$ corresponds to a quantum measurement.
\end{Corollary}
 
Next, we focus on the $M\geq 3$ case in which the operator set $\{O_m\}$ in $F(U)$ describes a quantum measurement. 

\begin{figure}
	\begin{center}
		\includegraphics[width= 0.4\columnwidth]{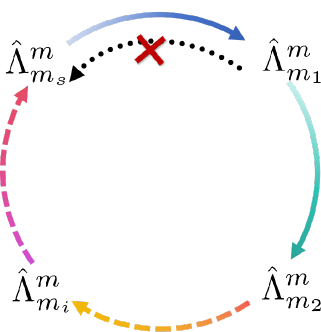}
	\end{center}
	\caption{The element exchange among different blocks of $\hat{\sigma}_m$ in Eq.~(\ref{mainform}) induced by simultaneous critical points. $\hat{\Lambda}_i^m\rightarrow\hat{\Lambda}_j^m$ denotes the event that at least one diagonal element of the block $\hat{\Lambda}_i^m$ is transferred to the block $\hat{\Lambda}_i^m$, and it is proven in Theorem~\ref{theo3} that FTs emerge on the optimization landscape if and only if there exist critical points which induce directed cycles.}
	\label{fig3}
\end{figure}
 
\subsubsection{Projective measurements}

Assume that $O_m$s in the objective function~(\ref{linear}) form a projective measurement, i.e., $O_mO_{m'}=\delta_{mm'}O_m$ and $\sum_m O_m=\mathbb{I}_D$. It is easy to obtain that each $O_m$ is diagonalized into $\hat{O}_m$ as Eq.~(\ref{mainoform}), with the dominant block $\hat{\Omega}_m^m=\mathbb{I}_{d_m^\prime}$ and the rest being zero matrices $\hat{\Omega}_{m'}^m=\boldsymbol{0}_{d^\prime_{m'}}$, and an upper bound is obtained as
\begin{align}
	F(U^\prime)=\sum_{m=1}^M \tr(\pi\hat{\sigma}_m\pi^\dagger\hat{O}_m)\leq \sum_{m=1}^M \tr(\hat{\Lambda}_m^m)
\end{align}
 for any simultaneous critical point $U^\prime=\hat{Q}^\dagger U^\dagger_q\pi U_p \hat{P}$. If $\hat{\sigma}_m$ and $\hat{O}_m$ share the same block structure, i.e., $d_m^\prime=d_m$ for $m=1,\cdots,M$, this bound is reached and the measurement is called optimal.

Note further that $U^\prime$ acting on $\omega\rho_m$ amounts to the permutation $\pi$ acting on $\hat{\sigma}_m$, thereby swapping diagonal elements either within the same block or across different blocks of $\hat{\sigma}_m$. Interestingly, we are able to show that the pattern of such inter-block transfers has close connections to the emergence of FTs.

\begin{Theorem}\label{theo3}
When each diagonalized state $\hat{\sigma}_m$ in the objective function~(\ref{linear}) shares the same block structure with the corresponding diagonalized operator $\hat{O}_m$, any local maximum point given by $\hat{Q}^\dagger U_q^\dagger\pi U_p \hat{P}$ is a FT on the optimization landscape if and only if the permutation $\pi$ induces a unidirectional cyclic element exchange among at least three blocks $\Lambda_{m_i}^m$ of $\hat{\sigma}_m$ as Eq.~(\ref{mainform}). 
\end{Theorem}

Denote by $\Lambda_i^m\rightarrow\Lambda_j^m$ the event that at least one diagonal element of the block $\Lambda_i^m$ is transferred to the block $\Lambda_{j}^m$ and by $L=(\overrightarrow{m_1,\cdots,m_s})$ a closed loop $\Lambda^{m}_{m_1}\rightarrow \Lambda^{m}_{m_2}\rightarrow\cdots\rightarrow\Lambda^{m}_{m_s}\rightarrow \Lambda^{m}_{m_1}$. The above theorem provides a geometric interpretation that an FT emerges on the landscape if and only if the critical point induces a directed cycle, as illustrated in Fig.~\ref{fig3}. Especially, there exist only two closed loops $(\overrightarrow{1,2,3})$ and $(\overrightarrow{1,3,2})$ in $M=3$, and we have
\begin{Corollary}\label{corollary2}
False traps emerge on the optimization landscape of maximizing $F(U)$ in~(\ref{linear}) with $M=3$, if and only if there exists a simultaneous critical point inducing a loop $L=(\overrightarrow{m_1,m_2,m_3})$ such that
	\begin{equation}\label{cyclic}
		\lambda_{k_i}^{m_i}-\lambda_{k_i}^{m_{i+1}}\leq \lambda_{k_{i-1}}^{m_i}-\lambda_{k_{i-1}}^{m_{i+1}},~ i=1,2,3,
	\end{equation}
	with $m_0=m_3, m_4=m_0, k_i=\arg\min_{k\in\mathcal{I}_{m_i}}\lambda_k^{m_i}-\lambda_k^{m_i+1}$, and $\mathcal{I}_m=[\sum_{m'=1}^{m-1}d_{m'}+1,\sum_{m'=1}^md_{m'}]$.
\end{Corollary}

The proofs of Theorem~\ref{theo3} and Corollary~\ref{corollary2} are deferred to Appendix~\ref{appendixC1} and~\ref{appendixC2}. It is pointed out that there are possibly many FTs, leading to different landscape values. To illustrate the non-uniqueness of FTs, consider another example
\begin{eqnarray*} 
	\omega_1&=&\omega_2=\omega_3=1/3, \\
	\rho_1&=&{\rm diag}\{ 0.23,0.35,0.17,0.25,0,0\},\\
	\rho_2&=&{\rm diag}\{0,0,0.27,0.3,0.22,0.21\}, \\
	\rho_3&=&{\rm diag}\{ 0.15,0.26,0,0,0.35,0.24 \},\\
	O_1&=&{\rm diag}\{  1,1,0, 0,0,0\},\\
	O_2&=&{\rm diag}\{  0, 0,1,1,0,0\},\\
	O_3&=&{\rm diag}\{  0, 0,0,0,1,1\},
\end{eqnarray*} 
which admits a global maximum value $0.58$. One can verify that both $U_1=\pi(4,6)\pi(1,4)$ and $U_2=\pi(3,5)\pi(2,3)U_1$ are local maximum points, however, they lead to different landscape values, $F(U_1)=79/150$ and $F(U_2)=0.42$, thus forming two distinct FTs.

\subsubsection{General measurements}

If $O_m$s are elements of a POVM, then it follows directly from the Naimark dilation theorem that by coupling to an $M$-dimensional ancillary system, $F(U)$ can be purified to $\sum_{m=1}^M \omega_m\tr\left[\bar{U}\bar{\rho}_m\bar{U}^\dagger \bar{O}_m\right],$ where $\bar{\rho}_m=\rho_m\otimes |0\rangle\langle 0|$, operators $\bar{O}_m=\mathbb{I}_N\otimes |m\rangle\langle m|$ are elements of a projective measurement and satisfy $\sum_m \bar{O}_m=\mathbb{I}_{DM}$, and the extended unitary ansatz $\bar{U}=(U\otimes \mathbb{I}_M)V$ with $V(\rho_m\otimes|0\rangle\langle 0|)V^\dagger=\sum_{m'}\sqrt{O_{m'}}\rho_m\sqrt{O_{m'}}\otimes |m'\rangle\langle m'|$. It is obvious that the attainable $\bar{U}$ constitute only a subset of $\mathcal{U}(DM)$, thus leading a constrained optimization problem which violates Assumptions~\ref{assm1} and~\ref{assm2}.

 For the above objective function with the constrained ansatz $\bar{U}$, the presence of FTs is a generic feature of the optimization landscape~\cite{Ge2022,Ge2021}, confirmed by the following example
\begin{eqnarray*}
	\omega_1&=&\omega_2=0.25,~~\omega_3=0.5,\\
	\rho_1&=&{\rm diag}\{ 1,0,0,0\},\\
	\rho_2&=&{\rm diag}\{ 0,1,0,0\},\\
	\rho_3&=&{\rm diag}\{0,0,0.8,0.2\},\\
	O_1&=&{\rm diag}\{0.3,0.2,0.4,0.3\},\\
	O_2&=&{\rm diag}\{0.35,0.45,0.25,0.5\},\\
	O_3&=&{\rm diag}\{0.35,0.35,0.35,0.2\}.
\end{eqnarray*}
It is easy to verify that both $U_1=\pi(2,4)\pi(1,4)\pi(1,3)$ and $U_2=\pi(1,4)$ are local maxima, but lead to different landscape values, $F(U_1)=0.4$ and $F(U_2)=0.3625$, implying that $U_2$ forms a FT on the landscape.

\subsection{Distinguishability v.s. false traps}

As FTs originate from incompatible spectral orderings of different objective terms, it naturally raises the question of whether such incompatibility can be avoided entirely.  A particular scenario is considered where the states and observables are perfectly distinguishable.

\begin{Definition}
 The states $\rho_m$s are perfectly distinguishable if they satisfy
\begin{align}\label{distinguish}
	\tr[\rho_m\rho_{m'}]=0, ~~\forall~ m \neq m'.
\end{align} 
Similarly, operators $O_m$s are perfectly distinguishable if $\tr[O_m O_{m'}]=0$ for $m\neq m'$. 
\end{Definition}

With the perfect distinguishability assumption, we can show

\begin{Theorem}\label{theo4}
	If both states and observables are perfectly distinguishable, i.e., $\tr[\rho_m\rho_{m'}]=0$ and $\tr[O_mO_{m'}]=0$ for any $m\neq m'$, then the optimization landscape of $F(U)$ in~(\ref{linear}) is devoid of false traps formed by simultaneous critical points. That is, any simultaneous critical point is either a globally optimal point or a saddle.
\end{Theorem} 

The proof is detailed in Appendix~\ref{appendixD}. Theorem~\ref{theo4} establishes a close connection between FT and indistinguishability that the presence of FTs formed by simultaneous critical points can be ascribed to the loss of distinguishability among states and/or operators. To further reveal how state indistinguishability affects the existence of FTs, consider the following example 
\begin{eqnarray}
	\omega_1&=&\omega_2=\omega_3=1/3, \nonumber\\
	\rho_1&=&\text{diag}\{ 1-\epsilon_1, \epsilon_1, 0 \},~ O_1=\text{diag}\{  1,0 ,0\}, \nonumber\\
	\rho_2&=&\text{diag}\{  0, 1-\epsilon_2,\epsilon_2\},~O_2=\text{diag}\{ 0,1,0\},\nonumber \\
	\rho_3&=& \text{diag}\{ \epsilon_3,0,1-\epsilon_3 \},~ O_3=\text{diag}\{0,0,1\}, \label{example}
\end{eqnarray}
with $\epsilon_i\in[0,1/2]$ for $i=1,2,3$. Obviously, the maximal value of $F$ is $1-(\epsilon_1+\epsilon_2+\epsilon_3)/3$. It follows from Eqs.~(\ref{distinguish})  and~(\ref{example}) that states are not distinguishable if and only if $\epsilon_i$s are nonzero, and the indistinguishability degree can be quantified by 
\begin{align}
	\epsilon:=\frac{1}{3}\sum_i\epsilon_i.
\end{align}
It follows also from Theorem~\ref{theo3} that simultaneous critical points form FTs if and only if the inequality
\begin{equation}\label{ineq}
	\epsilon_i+2\epsilon_{i+1}\geq1
\end{equation} 
holds for $i=1,2,3$ and $\epsilon_4=\epsilon_1$. This implies that the presence of FTs requires the state indistinguishability $\epsilon$ to be larger than $1/3$. 

Finally, the numerical evidence given in Appendix~\ref{appendixE} suggests a much stronger result that under perfect distinguishability, no FTs appear on the optimization landscape, even among non-simultaneous critical points, which is formulated as
\begin{Conjecture}\label{conjecture}
	If both states and operators are perfectly distinguishable, i.e., $\tr[\rho_m\rho_{m'}]=0$ and $\tr[O_mO_{m'}]=0$ for any $m\neq m'$, then the optimization landscape of $F(U)$ in~(\ref{linear}) is always devoid of false traps. 
\end{Conjecture}

\section{Discussions}\label{sec5}

We have obtained that FTs can emerge on optimization landscapes of the multi-term objective~(\ref{linear}), and their appearance is attributed to the loss of distinguishability among states and/or operators, in the highly overparameterized regime (i.e., the parameter number is sufficiently large). It immediately follows that sufficient parameterization alone is insufficient to guarantee trap-free landscapes beyond the single-term regime, and explains why the trap-free property is obtained for the $M=1$ case and also why increasing the tunable parameter number alone does not completely mitigate the phenomenon of FTs. 

Our results also reveal a fundamental  mechanism that the emergence of FTs is rooted in the incompatibility of the spectral ordering among different objective terms. This bears a close resemblance to frustration in condensed matter systems~\cite{Mosseri2008,Leon2010}, e.g., a triangular antiferromagnet, where each pair of spins energetically favors antiparallel alignment, yet the triangular geometry prevents all pairwise preferences from being simultaneously satisfied. As a result, the system is forced into a frustrated configuration. Similarly, different objective terms in the optimization problem favor mutually incompatible optimization directions such that the optimization landscape is forced to generate FTs.

Our results have other practical implications for quantum-classical optimization, in addition to designing algorithmic optimizers. As quantum indistinguishability plays a prominent role in the optimization landscape, an alternative way to mitigating FTs is to follow the problem-design principles. Taking quantum machine learning as an example, one can choose proper encoding maps, ancillary systems, and measurement designs to enhance distinguishability and hence to reshape the optimization landscape at the fundamental level. Compared to previous algorithmic strategies such as incorporating random perturbations or momentum into gradients~\cite{Bottou2018,Guo2024}, it has an advantage of not increasing optimization complexity while those algorithms generally require more executions of quantum ansatz and measurement, helpful for near-term quantum devices which suffer from limited coherence time and high operational costs. 

We then clarify that FT is conceptually distinct from another well-known optimization obstacle in VQAs, namely barren plateau (BP), which describes polynomially vanishing gradients with system dimensionality $D$~\cite{McClean2018} and generically arises due to high expressibility of the quantum ansatz~\cite{Harrow2009}. While BPs reflect global flatness of the landscape and constitute a major obstacle in high-dimensional systems, FTs reflect local structural incompatibility among different terms in the objective function and are more prominent in low- and intermediate-dimensional regimes. Thus, these two  represent distinct and complementary challenges for VQAs, and importantly, their distinct origins indicate that strategies effective against one may not necessarily mitigate the other.

Finally, our results obtained for a linear form~(\ref{linear}) can be  extended to the general case where each $f_m$s are monotonically increasing functions. In this scenario, the positive derivative $f^\prime_m$ preserves the sign structure of the gradient contributions in Theorem~\ref{theo1} and the Hessian in Theorem~\ref{theo2}. Consequently, the existence conditions for simultaneous critical points and false traps, as well as the trap-free property under perfect distinguishability, still hold. Again, the landscape complexity and the emergence of FTs are fundamentally governed by the multi-term structure and the distinguishability of states and observables.

 \section{Conclusion}\label{sec6}

We have investigated optimization landscapes of the objective functions~(\ref{fun1}) and~(\ref{Uloss}). In particular, a complete framework is first established to identify all critical points on the optimization landscape via Theorem~\ref{theo1} and to classify them via Theorem~\ref{theo2}, and also provides practical guidance in designing algorithmic optimizers for general objective functions. Then, with the functions $f_m$ assumed to be linear, it is shown in Theorem~\ref{theo3} that the optimization landscape can still encounter FTs formed by simultaneous critical points, yielding a negative answer to the open question of whether the trap-free property holds for the general case $M\geq 1$. Finally, a close connection is revealed that the emergence of FTs is attributed to the loss of distinguishability among states and/or operators in the objective function, by obtaining in Theorem~\ref{theo4} shows that FTs formed by simultaneous critical points are absent on the landscape with perfect distinguishability and further generalizing as Conjecture~\ref{conjecture} that FTs are always absent with the perfect distinguishability condition.

Our work deepens the understanding of quantum-classical optimization and paves ways to developing more efficient and reliable quantum algorithms. Future works are left to explore the prevalence of FTs under relaxed assumptions,to integrate distinguishability-aware design with existing methods to mitigate FTs and BPs, and also to experimentally verify the landscape properties on near-term quantum devices.

\begin{acknowledgments} 
	This work is financially supported by National Key R\&D Program of China No. 2025YFE0217200, Quantum Science and Technology-National Science and Technology Major Project No. 2023ZD0301400 and No. 2023ZD0300600, Guangdong Provincial Quantum Science Strategic Initiative No. GDZX2303007, Hong Kong Research Grant Council (RGC) No. 15213924, and the CAS AMSS-polyU Joint Laboratory of Applied Mathematics.
\end{acknowledgments}

\appendix

\section{Matrix vectorization and duplication matrices}\label{proof1}

Given an arbitrary matrix $X=(x)_{D\times D}\in \mathbb{R}^{D\times D}$, its vectorization is given by 
\begin{align}
	{\rm vec}(X)=(x_{11},\cdots, x_{1D},\cdots, x_{D1},\cdots,x_{DD})^\top \in \mathbb{R}^{D^2}.
\end{align}
If $X$ is symmetric, i.e., $X^\top=X$, it can be expressed in a more compact form via half-vectorization
\begin{align}
	{\rm vech}(X)&=(x_{11},\cdots, x_{D1},x_{22},\cdots, x_{D2},\cdots,x_{DD})^\top \nonumber \\
&	\in \mathbb{R}^{\frac{D(D+1)}{2}},
\end{align}
by dropping all elements $x_{ij}$ with $j>i$. Indeed, these two vectors obey the exact relation
\begin{align}
  {\rm vec}(X)=D_{\rm sy}{\rm vech}(X)
\end{align}
with the duplication matrix 
\begin{equation}
	D_{\rm sy}=\left[\begin{array}{cccc}
		D_{11}& \boldsymbol{0} &\cdots & \boldsymbol{0}\\
		D_{21} & D_{22} & \cdots & \boldsymbol{0}\\
		\vdots& \vdots & \ddots & \vdots\\
		D_{D1}& D_{D2}& \cdots & D_{DD}
	\end{array}
	\right] \in \mathbb{R}^{D^2\times \frac{D(D+1)}{2}},
\end{equation}
where diagonal blocks are $D_{ii}=\begin{bmatrix}
	\boldsymbol{0}_{(i-1)\times (D+1-i)}\\
	\mathbb{I}_{D+1-i}
\end{bmatrix}$ and off-diagonal blocks $D_{ij}\in\mathbb{R}^{D\times(D+1-j)}$ have entries
\begin{equation*}
	d_{ss'}=\begin{cases}
		1 & \text{if}~ s=j, s'=i-j+1  \\
		0 &\text{otherwise}
	\end{cases}.
\end{equation*}

If $X$ is antisymmetric, i.e., $X^\top=-X$, its vectorization is expressed as 
\begin{align}
	{\rm vech}(X)&=(x_{21},\cdots,x_{D1},x_{32},\cdots,x_{D2},\cdots,x_{D(D-1)})^\top \nonumber \\
	& \in \mathbb{R}^{\frac{D(D-1)}{2}},
\end{align}
where elements $x_{ij}$ with $j\geq i$ are dropped. Similarly, there exists the exact relation
\begin{align}
	{\rm vec}(X)=D_{\rm ay}{\rm vech}(X)
\end{align} 
with duplication matrix
\begin{equation}
	D_{\rm ay}=\left[\begin{array}{cccc}
		D_{11}& \boldsymbol{0} &\cdots & \boldsymbol{0}\\
		D_{21} & D_{22} & \cdots & \boldsymbol{0}\\
		\vdots& \vdots & \ddots & \vdots\\
		D_{D1}& D_{D2}& \cdots & D_{DD}
	\end{array}
	\right] \in \mathbb{R}^{D^2\times \frac{D(D-1)}{2}},
\end{equation}
where diagonal blocks are $D_{ii}=\begin{bmatrix}
	\boldsymbol{0}_{i \times (D-i)}\\
	\mathbb{I}_{D-i}
\end{bmatrix}$
and off-diagonal blocks $D_{ij}\in\mathbb{R}^{D\times(D-j)}$ have entries
\begin{equation*}
	d_{ss'}=\begin{cases}
		-1 & \text{if}~s=j, s'=i-j  \\
		0 &\text{otherwise}
	\end{cases}.
\end{equation*}

\section{The proof of Eq.~(\ref{Hess2})}\label{proof2}

Eq.~(\ref{Hess}) in the main text can be explicitly written as
\begin{align}
	h_{U}(A)&=\sum_{m=1}^{M\geq 1} \omega_m\tr[A\{U\rho_m U^\dagger (AO_m-O_m A)\}] \nonumber \\
	&=\sum_{m=1}^{M\geq 1} \omega_m {\rm vec}(A^\dagger)^\dagger {\rm vec}[U\rho_m U^\dagger (AO_m-O_m A)] \nonumber \\
	&=\sum_{m=1}^{M\geq 1} \omega_m {\rm vec}(A)^\dagger [{\rm vec}(U\rho_m U^\dagger AO_m)\nonumber \\
	&~~~~~~~~~~~~-{\rm vec}(U\rho_m U^\dagger O_m A)].
\end{align}
The first equality follows from $\tr[XY]=\tr[YX]$, the second from $\tr[X^\dagger Y]={\rm vec}(X)^\dagger{\rm vec}(Y)$, and the last from ${\rm vec}(X+Y)={\rm vec}(X)+{\rm vec}(Y)$, for matrices $X, Y$. Following further from the relation ${\rm vec}(XYZ)=(Z^\top\otimes X){\rm vec}(Y)=\mathbb{I}\otimes XY{\rm vec}(Z)$ yields
\begin{align}\label{quard}
	h_{U}(A)&=\sum_{m=1}^{M\geq 1} \omega_m {\rm vec}(A)^\dagger [O_m^\top\otimes U\rho_m U^\dagger {\rm vec}(A)\nonumber \\
	&~~~~~~~~-\mathbb{I}\otimes U\rho_m U^\dagger O_m {\rm vec}(A)] \nonumber \\
	&\equiv{\rm vec}(A)^\dagger T {\rm vec}(A).
\end{align}

Denote by $A_{\rm re}$ and $A_{\rm im}$ the real and imaginary parts of $A$, respectively. If $A$ is Hermitian, i.e., $A^\dagger=A$, then we have $A^\top_{\rm re}=A_{\rm re}$ and $A^\top_{\rm im}=-A_{\rm im}$, obeying  
\begin{equation}\label{half}
	{\rm vec}(A_{\rm re})=D_{\rm sy}{\rm vech}(A_{\rm re}),~{\rm vec}(A_{\rm im})=D_{\rm ay}{\rm vech}(A_{\rm im})
\end{equation}
with duplication matrices  $D_{\rm sy}$ and $D_{\rm ay}$ defined above. Combining Eq.~(\ref{half}) with Eq.~(\ref{quard}) gives rise to
\begin{align}
	& {\rm vec}(A)^\dagger T {\rm vec}(A) \nonumber \\
	=&[{\rm vec}(A_{\rm re})+i{\rm vec}(A_{\rm im})]^\dagger T[{\rm vec}(A_{\rm re})+i{\rm vec}(A_{\rm im})] \nonumber \\
	=&[D_{\rm sy}{\rm vech}(A_{\rm re})-i D_{\rm ay}{\rm vech}(A_{\rm im})]^\top T  \nonumber \\
	&[D_{\rm sy}{\rm vech}(A_{\rm re})+i D_{\rm ay}{\rm vech}(A_{\rm im})] \nonumber \\
	=&\begin{bmatrix}{\rm vech}(A_{\rm re})\\{\rm vech}(A_{\rm im})
	\end{bmatrix}^\top\begin{bmatrix}
		D_{\rm sy}^\top T D_{\rm sy} & i D_{\rm sy}^\top T D_{\rm ay}\\
		- iD_{\rm ay}^\top T D_{\rm sy} & D_{\rm ay}^\top T D_{\rm ay}
	\end{bmatrix}\begin{bmatrix}{\rm vech}(A_{\rm re})\\{\rm vech}(A_{\rm im})
	\end{bmatrix} \nonumber \\ 
	=&\begin{bmatrix}{\rm vech}(A_{\rm re})\\{\rm vech}(A_{\rm im})
	\end{bmatrix}^\top\begin{bmatrix}
		D_{\rm sy}^\top T_{\rm re}D_{\rm sy} & -D_{\rm sy}^\top T_{\rm im}D_{\rm ay}\\
		D_{\rm ay}^\top T_{\rm im} D_{\rm sy} & D_{\rm ay}^\top T_{\rm re} D_{\rm ay} 
	\end{bmatrix}\begin{bmatrix}{\rm vech}(A_{\rm re})\\{\rm vech}(A_{\rm im})
	\end{bmatrix} \nonumber \\
	\equiv& \boldsymbol{v}^\top H \boldsymbol{v}.
\end{align}
The fourth equality is obtained by ignoring the imaginary part of the real $h_U(A)$. Since Hermitian $A$ could be arbitrary in $h_U(A)$, all elements of vector $\boldsymbol{v}=\begin{bmatrix}{\rm vech}(A_{\rm re})\\{\rm vech}(A_{\rm im})\end{bmatrix}\in\mathbb{R}^{D^2}$ could be arbitrarily valued. 

\section{The landscape framework for $M=1$}\label{appendixA.3}

Suppose that the objective function $F(U)$ in the main text has a single quantum state $\rho$ and operator $O$. It follows immediately from Theorem~\ref{theo1} that the critical-point condition~(\ref{condition}) simplifies to
\begin{align}
[U\rho\,U^\dagger, O]=0.\label{c1}
\end{align} 

Further, denote by $P$ the unitary transformation that diagonalizes Hermitian $\rho$ into
\begin{align}
	\hat{\rho}=P\rho P^\dagger&={\rm diag}\{\underbrace{\lambda_1,\cdots,\lambda_1}_{p_1},\cdots,\underbrace{\lambda_r,\cdots,\lambda_r}_{p_r} \} \nonumber \\
	&\triangleq{\rm blkdiag}\{\lambda_1\mathbb{I}_{p_1},\cdots,\lambda_r\mathbb{I}_{p_r}\}, \label{eigrho}
\end{align}
where distinct eigenvalues of $\rho$ are arranged as $\lambda_1>\cdots>\lambda_r$ and their multiplicities $p_1,\cdots,p_r$ satisfy $\sum_i p_i=D$. ${\rm blkdiag}\{ A_1,A_2,\cdots,A_l\}$ describes the block-diagonal matrix composed of $A_i$ along its diagonal. Denote by $Q$ the unitary transformation for $O$ such that
\begin{align}
	\hat{O}=Q\,O\, Q^\dagger&={\rm diag}\{\underbrace{o_1,\cdots,o_1}_{q_1},\cdots,\underbrace{o_s,\cdots,o_s}_{q_s} \} \nonumber \\
	&\triangleq{\rm blkdiag}\{o_1\mathbb{I}_{q_1},\cdots,o_s\mathbb{I}_{q_s}\}, \label{eigo}
\end{align}
where $o_1>\cdots>o_s$ are distinct eigenvalues of $O$, with multiplicities $q_1,\cdots,q_s$ satisfying $\sum_i q_i=D$.  

It is easy to verify that Eq.~(\ref{c1}) is equal to
\begin{equation}\label{condition2}
	[V\hat{\rho} V^\dagger, \hat{O}]=0
\end{equation} 
with $V=Q U P^\dagger$. Then, all critical points belong to
\begin{align} \label{solution}
	\mathcal{S}=\{Q^\dagger U^\dagger_q\pi U_p P: U_q\in \mathcal{U}(\boldsymbol{q}), U_p\in \mathcal{U}(\boldsymbol{p}), \pi\in \mathcal{P}\}.
\end{align}
Here $\mathcal{U}(\boldsymbol{p})$ represents the product group $\mathcal{U}(p_1)\times\cdots\times\mathcal{U}(p_r)$ with $\mathcal{U}(p_i)$ the $p_i$-dimensional unitary group for $i=1,\dots,r$, $\mathcal{U}(\boldsymbol{q})=\mathcal{U}(q_1)\times\cdots\times\mathcal{U}(q_s)$ is similarly defined, and $\mathcal{P}$ is the permutation group. Indeed, noting that $U_{p}={\rm blkdiag}\{U_{p_1},\cdots,U_{p_r}\}$ preserves $\hat{\rho}$, i.e., $U_p\hat{\rho} U^\dagger_p=\hat{\rho}$, $U_{q}={\rm blkdiag}\{U_{q_1},\cdots,U_{q_s}\}$ preserves $\hat{O}$ with $U_q\hat{O} U^\dagger_q=\hat{O}$, and permutation $\pi$ preserves the diagonal structure of $\hat{\rho}$ and $\hat{O}$,  we obtain that the set in Eq.~(\ref{solution}) is a solution set for the condition~(\ref{condition2}).

Given an arbitrary critical point $U^\prime=Q^\dagger U_q^\dagger \pi U_pP$, it is straightforward to compute
\begin{align}
	F(U^\prime)=\tr[\hat{O}\pi\rho\pi^\dagger]=\sum_{k=1}^D\lambda_{(k)} o_k,
\end{align}
where $\lambda_{(k)}$ and $o_k$ refer to the $k$-th diagonal entry of $\pi \hat{\rho} \pi^\dagger$ and $\hat{O}$, respectively. And the second-order derivative~(\ref{Hess}) in the main text becomes
\begin{align}
	h_{U^\prime}(A)&=\tr[\pi\hat{\rho}\pi^{\dagger}(\hat{A}\hat{O}\hat{A}-\hat{O}\hat{A}^2)]\nonumber\\
	&= {\rm vec}(\hat{A})^\dagger (\hat{O}\otimes \pi\hat{\rho}\pi^\dagger-\mathbb{I}\otimes \pi\hat{\rho}\pi^\dagger \hat{O}){\rm vec}(\hat{A})\nonumber\\
	&\equiv {\rm vec}(\hat{A})^\dagger T {\rm vec}(\hat{A})
\end{align}
with $\hat{A}=U_q\,Q A\,Q^\dagger U_q^\dagger$ and
\begin{align} 
	T&={\rm diag}\{ \lambda_{(1)}(o_{1}-o_1), \lambda_{(2)}(o_{1}-o_2), \cdots,\lambda_{(D)}(o_1-o_D), \nonumber\\
	&\cdots,\lambda_{(1)}(o_{D}-o_{1}),\lambda_{(2)}(o_{D}-o_2),\cdots,\lambda_{(D)}(o_D-o_D)\}. \nonumber
\end{align}
Here all eigenvalues of $O$ in Eq.~(\ref{eigo}) are relabelled as $o_k$ for $k=1,\cdots, D$. Combining matrix $T$ with Eq.~(\ref{Hess2}) in the main text yields	
\begin{align}
	H&=
	\begin{bmatrix}
		D^\top_{\rm sy}T D_{\rm sy} & \boldsymbol{0}\\
		\boldsymbol{0} & D^\top_{\rm ay}TD_{\rm ay}
	\end{bmatrix} \nonumber \\
&=	\begin{bmatrix}
	{\rm diag}\{W_1,W_2,\cdots,W_D\} & \boldsymbol{0}\\
	\boldsymbol{0} & {\rm diag}\{V_1,V_2,\cdots,V_D\}
\end{bmatrix},  \label{Hmatrix}
\end{align}
 with
\begin{align*}
  W_k =&{\rm diag}\{ 0,-(\lambda_{(k)}-\lambda_{(k+1)})(o_k-o_{k+1}), \nonumber \\
 &\cdots,-(\lambda_{(k)}-\lambda_{(D)})(o_{k}-o_D)\}, \\
 V_k =&{\rm diag}\{-(\lambda_{(k)}-\lambda_{(k+1)})(o_k-o_{k+1}), \nonumber \\
 &\cdots,-(\lambda_{(k)}-\lambda_{(D)})(o_{k}-o_D)\}.
\end{align*} 
 Obviously, the eigenvalues of $H$ are given by
\begin{equation}\label{heigen}
	h_{kk'}=-(\lambda_{(k)}-\lambda_{(k')})(o_k-o_{k'}),~ \forall 1\leq k\leq k'\leq D.
\end{equation}

\section{Simultaneous critical points}\label{appendixB}

If states $\rho_m$ are pairwise commutative, i.e., $[\rho_m,\rho_{m'}]=0$ for $m,  m'=1,\cdots,M$, then they can be simultaneously diagonalized into
\begin{align}
	\sigma_m:=\omega_m\,P\rho_m\,P^\dagger={\rm diag}\{\lambda_1^m,\cdots,\lambda_D^m\}, m=1,\cdots, M, \nonumber
\end{align}
by one single unitary transformation $P$. Note that $\lambda^m_k$s are generically not in the descending order as Eq.~(\ref{eigrho}). 

The $k$-th diagonal element of $\sigma_m$ is said to be dominant if $\lambda_{k}^m\geq\lambda_{k}^{m'}$ holds for all $m'\neq m$. An extra assignment rule is imposed that if there happens $\lambda_{k}^m=\lambda_{k}^{m'}$ for two different blocks $\Lambda^m$ and $\Lambda^{m^\prime}$, this element is assigned to block $\Lambda^m$ with $m<m^\prime$. Then, we show that there exists a permutation $R$ transforming each $\sigma_m$ into
\begin{equation}\label{form}
	\tilde{\sigma}_m=R\sigma_m R^\dagger ={\rm blkdiag}\{\underbrace{\Lambda^{m}_{1}}_{d_1},\cdots,\underbrace{\Lambda^m_{i}}_{d_i}, \cdots, \underbrace{\Lambda^m_{M}}_{d_M} \},
\end{equation}
where $\Lambda_{i=m}^m$ contains all dominant entries of $\sigma_m$. 

Noting first that all permutations do not change but reorder entries of all matrices, we can always find a permutation $\pi_1\in\mathcal{P}(D)$ such that 
\begin{equation}
	\pi _1\sigma_1\pi_1^\dagger={\rm blkdiag}\{\underbrace{\Lambda_1^1}_{d_1}, *\},
\end{equation}
where the diagonal block $\Lambda_1^1$ is composed of dominant elements of $\sigma_1$ with dimensionality $d_1$, and another permutation $\pi_2={\rm blkdiag}\{\mathbb{I}_{d_1},\pi \}$, with $\pi \in\mathcal{P}(D-d_1)$, 
\begin{eqnarray}
	\pi_2\pi_1\sigma_2\pi_1^\dagger \pi_2^\dagger ={\rm blkdiag}\{ \underbrace{\Lambda^2_1}_{d_1},\underbrace{\Lambda_2^2}_{d_2},* \}
\end{eqnarray}
with a $d_1$-dimensional diagonal block $\Lambda^2_1$ and the dominant matrix $\Lambda_2^2$ of $\pi_1\sigma_2\pi_1^\dagger$. It is remarked that $\Lambda^2_2$ is also the dominant block of $\sigma_2$, and $\pi_2$ does not alter the dominate block $\Lambda_1^1$ of $\sigma_1$, i.e.,
\begin{align}
	\pi_2\pi_1\sigma_1\pi_1^\dagger \pi_2^\dagger =	\pi_2\,{\rm blkdiag}\{ \Lambda_1^1,* \}\,\pi_2^\dagger=\{ \Lambda_1^1,* \}. \label{invariance}
\end{align}

Applying sequentially permutations $\pi_1,\cdots, \pi_m={\rm blkdiag}\{\mathbb{I}_{\sum_{m^\prime}^{m-1}d_{m^\prime}}, \pi\in \mathcal{P}(D-\sum_{m^\prime}^{m-1}d_{m^\prime}) \}$ to $\sigma_m$ yields
\begin{equation}
	\pi_m\cdots \pi_1 \sigma_m \pi_1^\dagger\cdots\pi^\dagger_m={\rm blkdiag}\{\underbrace{\Lambda_1^m}_{d_1},\underbrace{\Lambda_2^m}_{d_2},\cdots,\underbrace{\Lambda_m^{m}}_{d_m}, * \},
\end{equation}
where $\Lambda_{i}^{m}$ refers to a $d_i$-dimensional diagonal block and $\Lambda_m^{m}$ is element-wise dominant of $\sigma_m$. As a consequence, applying permutation $R:=\pi_M\cdots\pi_1$ to  all $\sigma_m$s and using the invariance property~(\ref{invariance}) gives rise to the desired form in Eq.~(\ref{form}).

Then, all diagonal $\tilde{\sigma}_m$s can be faithfully represented as vectors $\boldsymbol{r}_m=(\lambda^m_{11},\cdots, \lambda^m_{d_Md_M})^\top \in \mathbb{R}^D$, which are arranged into a matrix 
\begin{align}
	G=\begin{bmatrix} \boldsymbol{r}^\top_1\\ \boldsymbol{r}^\top_2 \\\cdot \\\cdot \\\cdot \\ \boldsymbol{r}^\top_M \end{bmatrix}=\begin{bmatrix} \boldsymbol{c}_1, \boldsymbol{c}_2,\cdots, \boldsymbol{c}_D \end{bmatrix} \in \mathbb{R}^{M\times D},
\end{align}
where $\boldsymbol{c}_k\in\mathbb{R}^M$ with the $m$-th element being the $k$-th diagonal entry of $\tilde{\sigma}_m$. It follows from Eq.~(\ref{form}) that if indices $k$ and $k^\prime$ belong to different intervals $\mathcal{I}_I:=[\sum_{i=1}^{I-1}d_{i}+1,\sum_{i=1}^I d_{i}]$, then there must be $\boldsymbol{c}_k\neq \boldsymbol{c}_{k^\prime}$. Further, there exists a permutation $S_\rho={\rm blkdiag}\{ \hat{\pi}_1\in \mathcal{P}(d_1),\cdots,\hat{\pi}_M\in\mathcal{P}(d_M)\}$ such that 
\begin{align}\label{sform}
	\hat{\sigma}_m&=S_\rho \tilde{\sigma}_m S_\rho^\dagger ={\rm blkdiag}\{ \underbrace{\hat{\pi}_1 \Lambda_1^m\hat{\pi}_1^\dagger}_{d_1},\cdots, \underbrace{\hat{\pi}_M \Lambda_M^m\hat{\pi}_M^\dagger}_{d_M} \} \nonumber \\
	&= {\rm blkdiag}\{ \underbrace{\hat{\Lambda}_1^m}_{d_1},\cdots, \underbrace{\hat{\Lambda}_M^m}_{d_M} \},
\end{align}
where $\hat{\Lambda}_m^m:=\hat{\pi}_m\Lambda_m^m\hat{\pi}_m^\dagger$ remains element-wise dominant, and the corresponding matrix becomes
\begin{equation}\label{ssform}
	\hat{G}=[\underbrace{\hat{\boldsymbol{\lambda}}_1,\cdots,\hat{\boldsymbol{\lambda}}_1}_{p_1},\cdots,\underbrace{\hat{\boldsymbol{\lambda}}_r,\cdots,\hat{\boldsymbol{\lambda}}_r}_{p_r}],
\end{equation}
where $\hat{\boldsymbol{\lambda}}_1,\cdots,\hat{\boldsymbol{\lambda}}_r$ are distinct vectors reordered from $\boldsymbol{r}_1, \cdots, \boldsymbol{r}_M$, with multiplicities $p_1,\cdots, p_r$. 

Similarly, if operators $O_m$ are pairwise commutative, then they can be simultaneously transformed into 
\begin{equation}\label{oform}
	\hat{O}_m=S_oR_oQO_mQ^\dagger R_o^\dagger S_o^\dagger={\rm blkdiag}\{ \underbrace{\hat{\Omega}_1^m}_{d^\prime_1},\cdots,\underbrace{\hat{\Omega}_M^m}_{d^\prime_M}  \}
\end{equation} 
under a unitary transformation $Q$, permutation $R_o$, and $S_o$, with the dominant block $\hat{\Omega}_m^m$ of $\hat{O}_m$ for $m=1,\cdots, M$. Correspondingly, there is
\begin{equation}\label{ooform}
	\hat{G}_o=[\underbrace{\hat{\boldsymbol{o}}_1,\cdots,\hat{\boldsymbol{o}}_1}_{q_1},\cdots,\underbrace{\hat{\boldsymbol{o}}_s,\cdots,\hat{\boldsymbol{o}}_s}_{q_s}],
\end{equation}
where $\hat{\boldsymbol{o}}_1,\cdots,\hat{\boldsymbol{o}}_r$ are distinct vectors with multiplicities $q_1,\cdots, q_s$.

Thus, simultaneous critical points are given in
\begin{align} \label{solution2}
	\mathcal{S}=\{\hat{Q}^\dagger U^\dagger_q\pi U_p \hat{P}: U_q\in \mathcal{U}(\boldsymbol{q}), U_p\in \mathcal{U}(\boldsymbol{p}), \pi\in \mathcal{P}\}.
\end{align}
Here $\hat{P}=S_\rho R_\rho P$, $\hat{Q}=S_oR_o Q$, $\mathcal{U}(\boldsymbol{q})$ is the product group $\mathcal{U}(q_1)\times\cdots\times\mathcal{U}(q_s)$, $\mathcal{U}(\boldsymbol{p})=\mathcal{U}(p_1)\times\cdots\times\mathcal{U}(p_r)$, and $\mathcal{P}$ is the permutation group. 

For any simultaneous critical point $U^\prime=\hat{Q}^\dagger U^\dagger_q\pi U_p \hat{P}$, we compute
\begin{equation}\label{afu}
	F(U^\prime)=\sum_{m=1}^M \tr(\pi\hat{\sigma}_m\pi^\dagger\hat{O}_m)=\sum_{m=1}^M\sum_{k=1}^D\lambda_{(k)}^mo^m_k,
\end{equation}
and
\begin{align}\label{D13}
	&h_{U^\prime}(A)=\sum_{m=1}^M \tr[\pi\hat{\sigma}_m\pi^\dagger(\hat{A}\hat{O}_m\hat{A}-\hat{O}_m\hat{A}^2)]\nonumber\\
	=&\sum_{m=1}^M\sum_{1\leq k<k'\leq D} -(\hat{\lambda}_{(k)}^m-\hat{\lambda}_{(k')}^m)(\hat{o}_{k}^m-\hat{o}_{k'}^m)(\hat{x}_{kk'}^2+\hat{y}_{kk'}^2)\nonumber\\
	=& \sum_{1\leq k<k'\leq D}-\langle \hat{\boldsymbol{\lambda}}_{(k)}- \hat{\boldsymbol{\lambda}}_{(k')}, \hat{\boldsymbol{o}}_k- \hat{\boldsymbol{o}}_{k'}\rangle (\hat{x}_{kk'}^2+\hat{y}_{kk'}^2)\nonumber\\
\triangleq& \sum_{1\leq k <k'\leq D} h_{kk'}(\hat{x}_{kk'}^2+\hat{y}_{kk'}^2).
\end{align}
Here, $\hat{\lambda}_{(k)}^m, \hat{o}_k^m$ are the $k$-th diagonal entries of $\pi\hat{\sigma}_m\pi^\dagger, \hat{O}_m$, respectively, $\hat{x}_{ij}, \hat{y}_{ij}$ are the real and imaginary parts of the $(i,j)$ element of $\hat{A}=U_q\,\hat{Q} A\,\hat{Q}^\dagger U_q^\dagger$, $\hat{\boldsymbol{\lambda}}_{(k)}$ is the $k$-th column of the compact matrix corresponding to $\hat{\sigma}_m$s, and the columns of $\hat{G}_o$ are relabelled as $\hat{\boldsymbol{o}}_k$ for $k=1,\cdots, D$. $\langle \cdot, \cdot\rangle$ denotes the inner product of two vectors.

It follows above that
\begin{align}
h_{kk'}=\sum_{m=1}^M-(\hat{\lambda}_{(k)}^m-\hat{\lambda}_{(k')}^m)(\hat{o}_{k}^m-\hat{o}_{k'}^m), \label{hkk}
\end{align}
generalizes Eq.~(\ref{heigen}) for $M=1$ to $M\geq 1$. Since $\tilde{A}$ runs over all possible Hermitian matrices, $h_{U^\prime}$ is always nonnegative (nonpositive) if and only if all $h_{kk'}$s are nonnegative (nonpositive), and equivalently, $U^\prime$ is a local minimum (maximum) if and only if all $h_{kk'}$ are nonnegative (nonpositive).

\section{Proofs of Theorem~\ref{theo3} and Corollary~\ref{corollary2}}\label{appendixC}

\subsection{Proof of Theorem~\ref{theo3}}\label{appendixC1}

\textit{Necessity}- Suppose that $U^\prime=\tilde{Q}^\dagger U^\dagger_q\pi^\prime U_p \tilde{P}$ is a local but not global maximum of $F(U)$. It follows above that the local maximum property of $U^\prime$ requires
\begin{align}
	h_{kk^\prime} \leq 0,~~\forall~~  k, k^\prime=1,\cdots, D.
\end{align}
By contradiction, $U^\prime$ is assumed to permute elements between two distinct blocks $\hat{\Lambda}_i^m$ and $\hat{\Lambda}_j^m$ in each $\sigma_m$. That is, $\pi^\prime$ rearranges $\hat{\lambda}^{m}_{i_1}$ from $\hat{\Lambda}_i^m$ as an element $\hat{\lambda}^{m}_{(j_1)}$ in $\hat{\Lambda}_j^m$ and $\lambda^{m}_{j_2}$ from $\hat{\Lambda}_j^m$ as $\lambda^{m}_{(i_2)}$ in $\hat{\Lambda}_i^m$. Consequently, we obtain
\begin{eqnarray*}
	h_{i_2j_1}&=&-\langle \hat{\boldsymbol{\lambda}}_{(i_2)}- \hat{\boldsymbol{\lambda}}_{(j_1)},\hat{\boldsymbol{o}}_{i_2}- \hat{\boldsymbol{o}}_{j_1}\rangle\\
	&=&-(\hat{\lambda}^{i}_{(i_2)}-\hat{\lambda}^{i}_{(j_1)})+(\hat{\lambda}^{j}_{(i_2)}-\hat{\lambda}^{j}_{(j_1)})\\
	&=& -(\hat{\lambda}^{i}_{j_2}-\hat{\lambda}^{i}_{i_1})+(\hat{\lambda}^{j}_{j_2}-\hat{\lambda}^{j}_{i_1}) \\
	&>& 0.
\end{eqnarray*}
The second equality follows from that $\hat{\lambda}_{i_1}^i\in\hat{\Lambda}_i^i$ and $\hat{\lambda}_{j_2}^j\in\hat{\Lambda}_j^j$ and the projective measurement $\{O_m\}$ is optimal, and the inequality from
$\lambda^{i}_{i_1}\geq \lambda^{j}_{i_1}$ and $\lambda^{j}_{j_2}> \lambda^{i}_{j_2}$. This contradicts the assumption that $U^\prime$ is a local maximum.

 Again by contradiction, $U^\prime$ is assumed to only reorder diagonal entries within each block $\Lambda_i^m$ in each $\sigma_m$, for $i, m=1,\cdots,M$. Then, we have
\begin{equation}
	F(U^\prime)=\sum_{m=1}^M \tr(\pi\hat{\sigma}_m\pi^\dagger\hat{O}_m)=\sum_{m=1}^M \tr(\Lambda_{m}^m)
\end{equation}
which achieves the maximal value and thus is globally optimal. This contradicts the assumption that $U^\prime$ is a FT. Thus, we complete the proof that the unidirectional element exchanges within $\hat{\sigma}$ induced by $U^\prime$ must happen among at least three blocks.

\textit{Sufficiency}- Assume that a local maximum $U^\prime=\tilde{Q}^\dagger U^\dagger_q\pi U_p \tilde{P}$ induces that diagonal elements undergo unidirectional exchange among blocks $\hat{\Lambda}_{m_i}^m$, $i=1,\cdots,s\geq 3$. Thus, we are able to obtain
\begin{align}\label{E2}
	F(U^\prime)&=\sum_{m=1}^M \sum_{k\in\mathcal{I}_m} \hat{\lambda}_{(k)}^m =\sum_{m=1}^M \sum_{k\in \mathcal{I}_m}\hat{\lambda}_k^{\hat{m}_k}\nonumber\\
	&\leq \sum_{m=1}^M \sum_{k\in\mathcal{I}_{m}}\hat{\lambda}_k^m= F^*,
\end{align}
where $\hat{m}_k$ refers to the index of $\mathcal{I}_{m'}$ which $\hat{k}$ belongs to (i.e., $\hat{k}\in\mathcal{I}_{\hat{m}_k}$). The inequality follows from the fact that for $k\in\mathcal{I}_m$, $\lambda_k^m\in\Lambda_m^m$ is dominant, and strictly holds that there must be $\lambda_k^{\hat{m}_k}<\lambda_k^{m_s}$ for at least one $k\in \mathcal{I}_{m_s}$. This proves that $U^\prime$ forms a FT.

\subsection{Proof of Corollary~\ref{corollary2}}\label{appendixC2}

\textit{Sufficiency}- If the condition~(\ref{cyclic}) in Corollary~\ref{corollary2} is satisfied, then one can verify that any simultaneous critical point $U^\prime=\hat{Q}^\dagger U^\dagger_q\pi U_p \hat{P}$ with $\pi=\pi(k_1,k_3)\pi(k_1,k_2)$  is a locally but not globally maximum point.

\textit{Necessity}- It follows directly from Theorem~\ref{theo3} that if FTs exist, there exists a loop $(\overrightarrow{m_1,m_2,m_3})$ such that at least one diagonal element of $\Lambda_{i}^m$ is permuted to some position associated with the block $\Lambda_{i+1}^m$. Furthermore, the local optimality of simultaneous critical points requires
\begin{eqnarray*}
	h_{kk'}=\begin{cases}
		0 & k,k'\in\mathcal{I}_m\\
		\hat{\lambda}^{m'}_{(k)}-\hat{\lambda}^{m}_{(k)}+\hat{\lambda}^{m}_{(k')}-\hat{\lambda}^{m'}_{(k')}\leq 0 & k\in\mathcal{I}_m, k'\in\mathcal{I}_{m'}
	\end{cases},
\end{eqnarray*}
and thus the condition~(\ref{cyclic}) is obtained as desired.

\section{Proof of Theorem~\ref{theo4}}\label{appendixD}

If both states and operators are perfectly distinguishable, there are $\hat{\Lambda}_{m'}^m=\boldsymbol{0}_{d_{m'}}$ and $\hat{\Omega}_{m'}^m=\boldsymbol{0}_{d^\prime_{m'}}$ for any $m'\neq m$. Without loss of generality, the diagonal entries of $\hat{\Omega}_m^m$ are supposed to be in descending order; Otherwise, one can always transform $\hat{O}_m$ into this form by using some permutation ${\rm blkdiag}\{\pi_1\in\mathcal{P}(d^\prime_1),\cdots, \pi_M\in\mathcal{P}(d^\prime_M) \}$. In the following, we prove that for any given monotonically increasing function $f_m$, the landscape is devoid of FTs formed by simultaneous critical points under the distinguishability condition.

\textit{Proof}-
It equals that any simultaneous local maximum must be globally maximal. For any $U^\prime=\hat{Q}^\dagger U_q^\dagger \pi U_p \hat{P}\in\mathcal{S}$ satisfying
\begin{equation}
	\pi \hat{\rho}_m\pi^\dagger={\rm diag}\{ \lambda_{(1)}^m,\cdots,\lambda_{(D)}^m  \}\triangleq{\rm blkdiag}\{\underbrace{\Gamma_{1}^m}_{d^\prime_1},\cdots,\underbrace{\Gamma_{M}^m}_{d^\prime_M} \},
\end{equation}
 its type is determined by signs of the coefficients 
\begin{eqnarray*}\label{eigen}
	h_{kk'}&=&-\langle \hat{\boldsymbol{\lambda}}_{(k)}-\hat{\boldsymbol{\lambda}}_{(k')},\hat{\boldsymbol{o}}_k-\hat{\boldsymbol{o}}_{k'}\rangle\nonumber
	\\
	&=&-\omega_{m_1}(U^\prime)(\lambda^{m_1}_{(k)}-\lambda^{m_1}_{(k')})(o^{m_1}_k-o^{m_1}_{k'})-\omega_{m_2}(U^\prime)\cdot\\
	&&(\lambda^{m_2}_{(k)}-\lambda^{m_2}_{(k')})(o^{m_2}_k-o^{m_2}_{k'}),
\end{eqnarray*}
for $1\leq k<k'\leq D$, where $m_1$ and $m_2$ are subscript numbers of intervals $\mathcal{I}_m$ in which $i$ and $j$ belong to, respectively. If $U^\prime$ is a local maximum point, i.e., $h_{kk'}\leq0$, there must be $\lambda^{m_1}_{(k)}\geq\lambda^{m_1}_{(k')}$ and $\lambda^{m_2}_{(k)}\leq\lambda^{m_2}_{(k')}$, due to $\lambda^{m}_{(i)}\lambda^{m'}_{(i)}=o^m_io^{m'}_i=0$ for any $m\neq m'$ and $i=1,\cdots,D$. This indicates that diagonal entries of the block $\Gamma_{m}^{m}$ are the first $d^\prime_m$ largest eigenvalues of $\hat{\rho}_m$ and arranged in descending order, leading to a globally maximal landscape value $F_{{\rm max}}=\sum_{m=1}^M f_m\left(\sum_{s=1}^{d^\prime_m}\lambda_{ms}\right)$, where $\lambda_{ms}$ denotes the $s$-th largest eigenvalue of $\hat{\rho}_m$. Thus, the local maximum point is globally maximal. Similarly, if $U^\prime$ is a local minimum point, it follows from $h_{kk'}\geq0$ that the diagonal entries of $\Gamma_m^{m}$ are the first $d^\prime_m$ smallest eigenvalues of $\hat{\rho}_m$ and arranged in increasing order. This immediately implies that the local minimum point corresponds to a globally minimal landscape value, and hence is globally minimal.

\section{Numerical simulations for the perfect distinguishability case}\label{appendixE}

In the case that both states $\{\rho_m\}$ and operators $\{O_m\}$ are perfectly distinguishable, it has been proven that there exist no false traps formed by simultaneous critical points. As non-simultaneous critical points are widely spread over the landscape, it is necessary to further explore whether spuriously non-simultaneous optima can exist. To this end, we perform large-scale numerical simulations, and empirical evidence suggests that FTs are likely absent or at least rarely emerge.

\begin{figure}
	\begin{center}
		\includegraphics[width= \columnwidth]{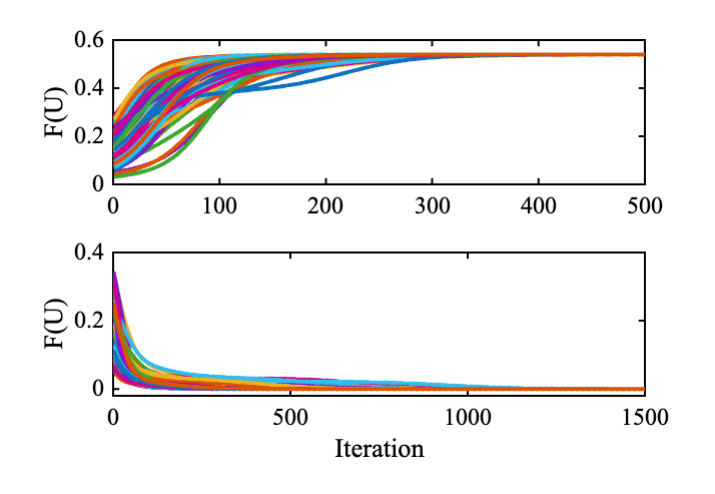}
	\end{center}
	\caption{The numerical simulation of gradient-ascent and gradient-descent searches for the example in Appendix~\ref{appendixE}. 100 independent simulations are performed with randomly generated initial seeds, and the results demonstrate that the maximal or minimal landscape values can be always achieved.}
	\label{fig4} 
\end{figure}

To clarify our findings, we here present an explicit example, in which $p_1=p_2=1/4, p_3=1/2$, and
\begin{align*}
	\rho_1&={\rm diag}\{1, 0, 0,0  \},~	O_1={\rm diag}\{0.8,0.2, 0, 0  \},\\
	\rho_2&={\rm diag}\{ 0, 1,0,0\},~O_2={\rm diag}\{0, 0, 0.4, 0\},\\
	\rho_3&={\rm diag}\{0, 0,0.8,0.2 \},~O_{3}={\rm diag}\{ 0,0,0,0.6 \}.
\end{align*}
One can easily verify that the minimal and maximal landscape values of $F(U)$ are $F_{\min}=0$ and $F_{\max}=0.54$, respectively. Fig.~\ref{fig4} shows that some global maximum or minimum point can be always achieved via the gradient-ascent or gradient-descent method, suggesting the absence of false traps. Moreover, using the Matlab Function `fsolve', we solve the critical-point condition with randomly generated initial seeds, and categorize the obtained critical points by their corresponding landscape value and characteristics (e.g., the type and simultaneity). As shown in Fig.~\ref{fig5}, the critical points derived from $10^5$ seeds yield distinct landscape values distributed between $0$ and $0.54$, and the vast majority of them are identified as saddle points. All of the found non-simultaneous critical points are identified as saddles. These results indicate that FTs are absent, or at least exceedingly rare, on the landscape. 

\begin{figure}
	\begin{center}
		\includegraphics[width= \columnwidth]{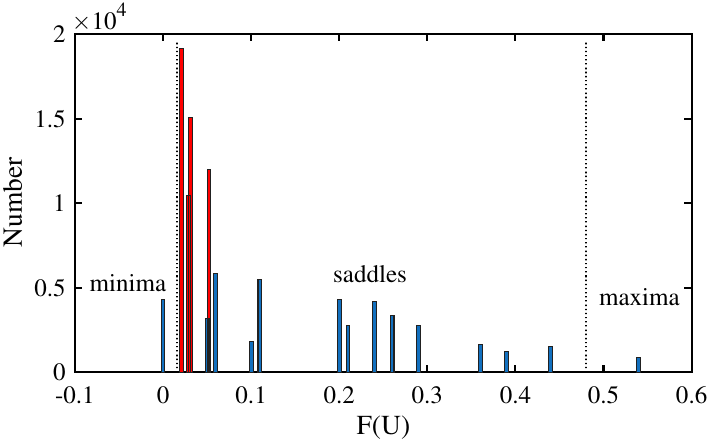}
	\end{center}
	\caption{Distribution of critical points categorized  according to their corresponding landscape values. The colourful bars indicates whether the critical points are simultaneous, with `blue' representing simultaneous and 'red' non-simultaneous. }
	\label{fig5} 
\end{figure}

\bibliography{landscape}

\end{document}